\documentclass[a4paper,UKenglish,cleveref, autoref, thm-restate.]{lipics-v2021}
\title{Intermediate Relation Size Bounds for Select-Project-Join-Union Query Plans} %TODO Please add

\author{Hubie Chen}{King's College London, UK }{hubie.chen@kcl.ac.uk}{https://orcid.org/0000-0000-0000-0000}{}%TODO mandatory, please use full name; only 1 author per \author macro; first two parameters are mandatory, other parameters can be empty. Please provide at least the name of the affiliation and the country. The full address is optional. Use additional curly braces to indicate the correct name splitting when the last name consists of multiple name parts.

\author{Markus Schneider}{King's College London, UK}{markus.schneider@kcl.ac.uk}{https://orcid.org/0000-0001-8338-5660}{}

\authorrunning{H. Chen and M. Schneider} %TODO mandatory. First: Use abbreviated first/middle names. Second (only in severe cases): Use first author plus 'et al.'

\Copyright{Hubie Chen and Markus Schneider} %TODO mandatory, please use full first names. LIPIcs license is "CC-BY";  http://creativecommons.org/licenses/by/3.0/

 %TODO mandatory; please add comma-separated list of keywords

\nolinenumbers %uncomment to disable line numbering

\EventEditors{John Q. Open and Joan R. Access}
\EventNoEds{2}
\EventLongTitle{42nd Conference on Very Important Topics (CVIT 2016)}
\EventShortTitle{CVIT 2016}
\EventAcronym{CVIT}
\EventYear{2016}
\EventDate{December 24--27, 2016}
\EventLocation{Little Whinging, United Kingdom}
\EventLogo{}
\SeriesVolume{42}
\ArticleNo{23}
\begin{document}

\newcommand{\from}[3]{\textbf{FROM} #1 \textbf{TO} #2: #3}
%Complexity classes
\def\FP{\text{\rm FP}}
\def\FL{\text{\rm FL}}
\def\spanL{\text{\rm SpanL}}
\def\spanLL{\text{\rm SpanLL}}
\def\sharpP{\text{\rm \#P}}
\def\sharpL{\text{\rm \#L}}
\def\NL{\text{\rm NL}}
\def\PTIME{\text{\rm P}}
\def\PH{\text{\rm PH}}
\def\NP{\text{\rm NP}}
\def\LOGSPACE{\text{\rm L}}
\def\NSPACE{\rm NSPACE}
\def\NTIME{\rm NTIME}
\def\co{\rm co\text{-}}
\def\PSPACE{\rm PSPACE}
\def\EXPTIME{\rm EXPTIME}
\def\NEXP{\rm NEXPTIME}
\def\TWOEXPTIME{\rm 2EXPTIME}
\def\AEXSPACE{\rm AEXSPACE}
\def\ACZ{\rm AC_0}
\def\hard{\rm \text{-}hard}
\def\complete{\text{-{\rm complete}}}
\def\assign{ \text{:--} }

\newcommand\sem[1]{{[\![ #1 ]\!]}} %Semantics
\newcommand{\db}[1]{\mathsf{db}(#1)}
\newcommand{\depth}[1]{\mathsf{depth}(#1)}
\newcommand{\mi}[1]{\mathit{#1}}
\newcommand{\ins}[1]{\mathbf{#1}}
\newcommand{\adom}[1]{\mathsf{dom}(#1)}
\newcommand{\ra}{\rightarrow}
\newcommand{\fr}[1]{\mathsf{fr}(#1)}
\newcommand{\dep}{\Sigma}
\newcommand{\sch}[1]{\mathsf{sch}(#1)}
\newcommand{\esch}[1]{\mathsf{edb}(#1)}
\newcommand{\isch}[1]{\mathsf{idb}(#1)}
\newcommand{\sign}{\ins{S}}
\newcommand{\body}[1]{\mathsf{body}(#1)}
\newcommand{\head}[1]{\mathsf{head}(#1)}
\newcommand{\guard}[1]{\mathsf{guard}(#1)}
\newcommand{\class}[1]{\mathsf{#1}}
\newcommand{\pos}[1]{\mathsf{pos}(#1)}
\newcommand{\crel}[1]{\prec_{#1}}
\newcommand{\base}[1]{\mathsf{base}(#1)}
\newcommand{\var}[1]{\mathsf{var}(#1)}
\newcommand{\vr}[1]{\langle #1 \rangle}
\newcommand{\const}[1]{\mathsf{const}(#1)}

\newcommand{\lsign}[1]{\mathsf{b}(#1)}
\newcommand{\lvar}[1]{\mathsf{v}(#1)}

\newcommand{\reach}[1]{\rightsquigarrow_{#1}}
\newcommand{\obl}{\mathsf{o}}
\newcommand{\sobl}{\mathsf{so}}
\newcommand{\std}{\mathsf{std}}
\newcommand{\cta}[1]{\class{CT}_{\forall \forall}^{#1}}
\newcommand{\cte}[1]{\class{CT}_{\forall \exists}^{#1}}

\newcommand{\ctda}[2]{\class{CT}_{\forall,#2}^{#1}}
\newcommand{\ctde}[2]{\class{CT}_{\exists,#2}^{#1}}

\newcommand{\ct}[1]{\class{CT}_{\forall}^{#1}}
\newcommand{\ctd}[2]{\class{CT}^{#1}_{#2}}
\newcommand{\ctapr}[1]{\mathsf{CT}_{\forall \forall}^{#1}}
\newcommand{\ctepr}[1]{\mathsf{CT}_{\forall \exists}^{#1}}
\newcommand{\ctdapr}[1]{\mathsf{CT}_{\forall}^{#1}}
\newcommand{\ctdepr}[1]{\mathsf{CT}_{\exists}^{#1}}
\newcommand{\ctpr}[1]{\mathsf{CT}_{\forall}^{#1}}
\newcommand{\ctdpr}[1]{\mathsf{CT}^{#1}}
\newcommand{\cri}[1]{\mathsf{cr}(#1)}
\newcommand{\lin}[1]{\mathsf{Lin_{\class{S}}}(#1)}
\newcommand{\ling}[1]{\mathsf{lin}(#1)}
\newcommand{\shape}[1]{\mathsf{shape}(#1)}
\newcommand{\svar}[1]{\mathsf{svar}(#1)}
\newcommand{\constfree}[1]{\mathsf{c\text{-}free}(#1)}
\newcommand{\id}[2]{\mathsf{id}_{#1}(#2)}
\newcommand{\unique}[1]{\mathsf{unique}(#1)}
\newcommand{\simple}[1]{\mathsf{simple}(#1)}
\newcommand{\gsimple}[1]{\mathsf{gsimple}(#1)}
\def\sub{\sqsubseteq}
\def\substrict{\sqsubset}

\newcommand{\rew}[1]{\hat{#1}}
\newcommand{\whichver}[1]{#1^w}
\newcommand{\support}[1]{\mathsf{support}(#1)}
\newcommand{\bagsupport}[1]{\mathsf{bagsupport}(#1)}
\newcommand{\norm}[1]{\mathsf{Norm}(#1)}
\newcommand{\depg}[1]{\mathsf{dg}(#1)}
\newcommand{\edepg}[1]{\mathsf{edg}(#1)}
\newcommand{\sodg}[1]{\mathsf{so\text{-}dg}(#1)}
\newcommand{\odg}[1]{\mathsf{o\text{-}dg}(#1)}
\newcommand{\ex}[1]{\mathsf{exvar}(#1)}
\newcommand{\mgu}[2]{\mathsf{mgu}(#1,#2)}
\newcommand{\res}[1]{\mathsf{res}(#1)}
\newcommand{\f}[2]{f_{#1}(#2)}
\newcommand{\arity}[1]{\mathsf{ar}(#1)}
\newcommand{\atoms}[1]{\mathsf{atoms}(#1)}
\newcommand{\lfacts}[1]{\mathsf{LFacts}(#1)}
\newcommand{\why}[3]{\mathsf{why}(#1,#2,#3)}
\newcommand{\trio}[3]{\mathsf{trio}(#1,#2,#3)}
\newcommand{\which}[3]{\mathsf{which}(#1,#2,#3)}
\newcommand{\whichnr}[3]{\mathsf{which}_{\mathsf{NR}}(#1,#2,#3)}
\newcommand{\whichun}[3]{\mathsf{which}_{\mathsf{UN}}(#1,#2,#3)}
\newcommand{\whichstar}[3]{\mathsf{which}_{\star}(#1,#2,#3)}
\newcommand{\posbool}[3]{\mathsf{posbool}(#1,#2,#3)}
\newcommand{\posboolnr}[3]{\mathsf{posbool}_{\mathsf{NR}}(#1,#2,#3)}
\newcommand{\posboolun}[3]{\mathsf{posbool}_{\mathsf{UN}}(#1,#2,#3)}
\newcommand{\bx}[3]{\mathbb{B}(#1,#2,#3)}
\newcommand{\bxnr}[3]{\mathbb{B}_{\mathsf{NR}}(#1,#2,#3)}
\newcommand{\bxun}[3]{\mathbb{B}_{\mathsf{UN}}(#1,#2,#3)}
\newcommand{\nx}[3]{\mathbb{N}(#1,#2,#3)}
\newcommand{\nrwhy}[3]{\mathsf{why}_{\mathsf{NR}}(#1,#2,#3)}
\newcommand{\mdwhy}[3]{\mathsf{why}_{\mathsf{MD}}(#1,#2,#3)}
\newcommand{\mtd}[3]{\mathsf{min\text{-}tree\text{-}depth}(#1,#2,#3)}
\newcommand{\mgd}[3]{\mathsf{min\text{-}dag\text{-}depth}(#1,#2,#3)}
\newcommand{\unwhy}[3]{\mathsf{why}_{\mathsf{UN}}(#1,#2,#3)}
\newcommand{\nrtrio}[3]{\mathsf{trio}_{\mathsf{NR}}(#1,#2,#3)}
\newcommand{\untrio}[3]{\mathsf{trio}_{\mathsf{UN}}(#1,#2,#3)}
\newcommand{\precnode}[1]{\prec_{#1}}
\newcommand{\repr}[2]{\mathsf{repr}_{#2}(#1)}
\newcommand{\reprtree}[2]{\mathsf{reprTrees}^{#1}(#2)}
\newcommand{\childlabels}[1]{\mathsf{child\textrm{-}labels}(#1)}
\newcommand{\birth}[2]{\mathsf{birth}_{#1}(#2)}
\newcommand{\OMIT}[1]{}
\newcommand{\crt}[1]{\texttt{cr}(#1)}
\newcommand{\precd}[1]{\prec_{#1}}
\newcommand{\level}[2]{\mathsf{lvl}_{#2}(#1)}
\newcommand{\lvl}[2]{\texttt{lv}_{#1}(#2)}
\newcommand{\posvar}[2]{\mathsf{pos}(#1,#2)}
\newcommand{\posterm}[2]{\mathsf{pos}(#1,#2)}
\newcommand{\varpos}[2]{\mathsf{var}(#1,#2)}
\newcommand{\termpos}[2]{\mathsf{term}(#1,#2)}
\newcommand{\CT}[2]{\mathsf{CT}^{#1}_{#2}}
\newcommand{\frpos}[1]{\mathsf{frpos}(#1)}
\newcommand{\dom}{\mathbf{C}}
\newcommand{\freshdom}{\mathbf{N}}
\newcommand{\frontier}[1]{\mathsf{fr}(#1)}
\newcommand{\eqtype}[1]{\mathsf{eqtype}(#1)}

\newcommand{\whymult}[3]{\mathsf{whymult}_{\mathsf{ALL}}(#1,#2,#3)}
\newcommand{\nrwhymult}[3]{\mathsf{whymult}_{\mathsf{NR}}(#1,#2,#3)}
\newcommand{\unwhymult}[3]{\mathsf{whymult}_{\mathsf{UN}}(#1,#2,#3)}
\newcommand{\vecenc}[1]{\mathsf{vec}(#1)}
\newcommand{\ground}[2]{\mathsf{gg}(#1,#2)}
\newcommand{\groundh}[2]{\mathsf{gh}(#1,#2)}
\newcommand{\multiset}[1]{\{\!\{ #1 \}\!\}}

\def\iso{\simeq}
\newcommand{\can}[1]{\mathsf{can}(#1)}
\newcommand{\proj}[2]{\Pi_{#1}(#2)}
\newcommand{\pred}[1]{\mathit{pred}(#1)}
\newcommand{\predt}[1]{[#1]}
\newcommand{\atom}[1]{\underline{#1}}
\newcommand{\tuple}[1]{\bar{#1}}
\newcommand{\resolv}[1]{[#1]}
\newcommand{\parent}[1]{\mathit{par}(#1)}
\newcommand{\dept}[1]{\mathit{depth}(#1)}
\newcommand{\chase}[2]{\mathsf{chase}(#1,#2)}
\newcommand{\chasesize}[2]{\mathsf{chsize}(#1,#2)}
\newcommand{\starchasei}[3]{\star\text{-}\mathsf{chase}^{#3}(#1,#2)}
\newcommand{\starchase}[2]{\star\text{-}\mathsf{chase}(#1,#2)}
\newcommand{\sochasei}[3]{\sobl\text{-}\mathsf{chase}^{#3}(#1,#2)}
\newcommand{\sochase}[2]{\sobl\text{-}\mathsf{chase}(#1,#2)}
\newcommand{\ochasei}[3]{\obl\text{-}\mathsf{chase}^{#3}(#1,#2)}
\newcommand{\ochase}[2]{\obl\text{-}\mathsf{chase}(#1,#2)}
\newcommand{\completion}[2]{\mathsf{complete}(#1,#2)}
\newcommand{\mar}[1]{\hat{#1}}
\newcommand{\nullobl}[3]{\bot^{#1}_{#2,#3}}
\newcommand{\nullsobl}[3]{\bot^{#1}_{#2,#3_{|\frontier{#2}}}}
\newcommand{\startype}[1]{\star\textrm{-}\mathsf{type}(#1)}
\newcommand{\type}[2]{\mathsf{type}_{#1}(#2)}
\newcommand{\types}[2]{#1\textrm{-}\mathsf{types}(#2)}
\newcommand{\src}[1]{\mathsf{src}(#1)}
\newcommand{\IDTGD}{\class{ID}}
\newcommand{\DLLITETGD}{\mathsf{DL\textrm{-}Lite^{TGD}}}
\newcommand{\SLTGD}{\class{SL}}
\newcommand{\LDAT}{\class{LDat}}
\newcommand{\NRDAT}{\class{NRDat}}
\newcommand{\GDAT}{\class{G}}
\newcommand{\WGTGD}{\class{WG}}
\newcommand{\RATGD}{\class{RA}}
\newcommand{\LARATGD}{\class{LARA}}
\newcommand{\LCRATGD}{\class{LCRA}}
\newcommand{\LCWATGD}{\class{LCWA}}
\newcommand{\WATGD}{\class{WA}}
\newcommand{\DAT}{\class{Dat}}
\newcommand{\UCQ}{\class{UCQ}}
\newcommand{\SLRATGD}{\class{SLRA}}
\newcommand{\SLWATGD}{\class{SLWA}}
\newcommand{\LRATGDP}{\class{LRA}^{+}}
\newcommand{\LWATGDP}{\class{LWA}^{+}}
\newcommand{\oblrew}[1]{\mathsf{enrichment}(#1)}

\def\eqtree{\approx}
\newcommand{\downof}[2]{#1_{\downarrow#2}}
\newcommand{\gri}[2]{\mathsf{gri}(#1,#2)}
\newcommand{\downc}[3]{\mathsf{down}(#1,#2,#3)}
\newcommand{\quot}[1]{#1_{/\eqtree}}
\newcommand{\leaves}[1]{\mathsf{leaves}(#1)}
\newcommand{\dgt}[1]{\mathsf{dag}(#1)}
\newcommand{\cq}[1]{\mathsf{cq}(#1)}
\newcommand{\cqs}[2]{\mathsf{cq}(#1,#2)}
\newcommand{\cqsu}[2]{\mathsf{cq}(#1,#2)}
\newcommand{\cqseq}[2]{\mathsf{cq}^{=}(#1,#2)}
\newcommand{\cqeq}[1]{\mathsf{cq}^{\approx}(#1)}
\newcommand{\cqeqs}[1]{\mathsf{cqs}^{=}(#1)}
\newcommand{\fo}[1]{\Psi_{#1}}
\newcommand{\fod}[1]{\Xi_{#1}}
\newcommand{\cnf}[3]{\phi_{#1,#2,#3}}
\newcommand{\trees}[1]{\mathsf{TR}(#1)}
\newcommand{\treesf}{\mathsf{TR}}
\newcommand{\utrees}[1]{\mathsf{UTR}(#1)}
\newcommand{\utreesf}{\mathsf{UTR}}
\newcommand{\spec}[1]{\mathsf{sp}(#1)}

\newcommand{\rela}{\mathbf{A}}
\newcommand{\relb}{\mathbf{B}}
\newcommand{\relc}{\mathbf{C}}
\newcommand{\reld}{\mathbf{D}}
\newcommand{\ar}{\mathrm{ar}}

\renewcommand{\H}{\mathbf{H}}
\newcommand{\tup}[1]{\overline{#1}}

\newcommand{\out}{\mathsf{out}}
\newcommand{\homs}{\mathsf{homs}}
\newcommand{\subplans}{\mathsf{subplans}}

\newcommand{\Q}{\mathbb{Q}}
\newcommand{\N}{\mathbb{N}}

\newcommand{\aug}{\mathsf{aug}}
\newcommand{\wid}{\textup{-width}}

\newcommand{\new}[1]{\textcolor{red}{#1}}    %{#1}    % {\textcolor{red}{#1}}

%%%%%%%%%%%%%%%%%%%%%%%%% Environment delimiters

\def\qed{\hfill{\qedboxempty}      % qed with empty box
  \ifdim\lastskip<\medskipamount \removelastskip\penalty55\medskip\fi}

\def\qedboxempty{\vbox{\hrule\hbox{\vrule\kern3pt
                 \vbox{\kern3pt\kern3pt}\kern3pt\vrule}\hrule}}

\def\qedfull{\hfill{\qedboxfull}   % qed with full box
  \ifdim\lastskip<\medskipamount \removelastskip\penalty55\medskip\fi}

\def\qedboxfull{\vrule height 4pt width 4pt depth 0pt}

\newcommand{\markfull}{\qedboxfull}
\newcommand{\markempty}{\qed} 

\newtheorem{manualtheoreminner}{Theorem}
\newenvironment{manualtheorem}[1]{%
  \renewcommand\themanualtheoreminner{#1}%
  \manualtheoreminner
}{\endmanualtheoreminner}

\newtheorem{manualpropositioninner}{Proposition}
\newenvironment{manualproposition}[1]{%
  \renewcommand\themanualpropositioninner{#1}%
  \manualpropositioninner
}{\endmanualpropositioninner}

\newtheorem{manuallemmainner}{Lemma}
\newenvironment{manuallemma}[1]{%
  \renewcommand\themanuallemmainner{#1}%
  \manuallemmainner
}{\endmanuallemmainner}

\newcommand{\triox}{\mathsf{Trio}[X]}

\newcommand{\mtodo}[1]{{\color{blue} Markus Todo: #1}}
\newcommand{\htodo}[1]{{\color{blue} Hubie Todo: #1}}

\maketitle

\begin{abstract}
    We study the problem of statically optimizing select-project-join-union (SPJU) plans where unary key constraints are allowed. A natural measure of a plan, which we call the output degree and which has been studied previously, is the minimum degree of a polynomial bounding the plan's output relation, as a function of the input database's maximum relation size. This measure is, by definition, invariant under passing from a plan to another plan that is semantically equivalent to the first. In this article, we consider a plan measure which we call the intermediate degree; this measure is defined to be the minimum degree of a polynomial bounding the size of all intermediate relations computed during a plan's execution --- again, as a function of the input database's maximum relation size.
    We present an algorithm that, given an SPJU plan $p$ and a set $\Sigma$ of unary keys, computes an SPJU plan $p'$ that is semantically equivalent to $p$ (over databases satisfying $\Sigma$) and that has the minimum intermediate degree over all such semantically equivalent plans. For the types of plans considered, we thus obtain a complete and effective understanding of intermediate degree.
\end{abstract}

\section{Introduction}

{\bf Background and motivation.}
Modern database systems capable of evaluating queries
of the standard query language SQL typically have,
as a key component, a query optimizer that synthesizes
execution plans consisting of basic operations from a 
relational algebra; techniques for estimating and comparing the 
costs of evaluating such plans are considered central
in the study of databases~\cite[Chapter 6]{AbiteboulHullVianu95-foundationsdatabases}.  
In this article, we study a basic measure 
of execution plans:
the \emph{intermediate size}, by which we mean
the maximum size 
over all intermediate relations computed during a query's execution,
%over all relations computed by a plan,
relative to a database.
We view this as a natural and fundamental measure:
the intermediate size gives an obvious lower bound on the amount of time
needed to execute a plan, for algorithms that explicitly materialize
all intermediate relations; and, it characterizes the amount of space
needed to store the results of all such intermediate relations.

We study this size measure in an asymptotic setting,
as a function of the maximum relation size~$M_\reld$ of a database $\reld$.
We say that a plan 
has \emph{intermediate degree} $d$ when $d \geq 0$ and
there exists a degree $d$ polynomial upper bounding the plan's intermediate size,
as a function of $M_{\reld}$, 
but no lower degree polynomial provides such a bound.
%(Clearly, a plan's output degree is 
%less than or equal to the plan's intermediate degree.)
We study this measure for SPJU query plans,
which allow \emph{select}, \emph{project}, \emph{join}
and \emph{union} 
as operations on relations.  These plans 
are semantically equivalent to unions of conjunctive queries,
and
are intensely studied throughout database theory.
We permit \emph{unary key constraints} to be expressed on database relations.
%basic relations in these plans. 
In essence, a unary key constraint states
that there is a single coordinate of a relation where, for each tuple in the relation, knowing the value at the coordinate determines the entire tuple.

Having identified a relevant complexity measure of query plans---the intermediate degree---and a relevant class of query plans---the SPJU query plans, one of the most basic and natural questions one could ask, from the standpoint of query optimization, is this: is there an algorithm that receives a query plan as input, and reformulates it into a semantically equivalent query plan whose intermediate degree is minimal, over all such equivalent plans?  That is, is there an algorithm that can statically analyze a given query plan, and reformulate it so as to optimally minimize the intermediate degree, without disrupting the plan's semantics?  A primary contribution of this article is a positive answer to this question.

{\bf Contributions.}
Our main theorem presents a minimization algorithm that, given an SPJU plan~$p$ 
with a set $\Sigma$ of unary key constraints,
computes an SPJU plan~$p'$ that is semantically equivalent to $p$ (with respect to databases satisfying the constraints $\Sigma$),
and has the lowest possible intermediate degree
over all SPJU plans that are so semantically equivalent to $p$.
%---along with the intermediate degree of $p'$.
This theorem thus yields, in the asymptotic setting considered,
a full understanding of how to
statically optimize the intermediate degree so that it is the minimum possible.
We remark that (to our knowledge) this type of understanding was not previously known even for SPJ plans (SPJU plans that do not use the union operator) and on databases without any constraints.
%for vanilla SPJ plans without any constraints.
We view the definition of the 
%asymptotic 
notion of \emph{intermediate degree},
and its supporting theory, as conceptual contributions of this work.

Our result showing how
to minimize the intermediate degree amounts to demonstrating how to compute, 
for any given plan,
an asymptotically optimal procedure
in a restricted model of computation, namely, the class of SPJU plans.
Our focus on polynomial degree is aligned with the notion of
fixed-parameter tractability, whereby a problem is tractable essentially
when there is a degree $d$ such that each so-called \emph{slice} 
has an algorithm with running time bounded above by a degree $d$ polynomial;
a \emph{slice} is a set of all instances sharing the same \emph{parameter},
which is a value associated with each instance.

To prove our result, we make a number of technical contributions that we believe will be of significant utility in the future study of intermediate degree, which
%, again, 
we view as a highly natural measure of plan complexity. 
One tool that we present (in Subsection~\ref{ssect:query-structures}) is 
a theorem showing how to convert an SPJ query plan to 
a relational structure along with a certain form of tree decomposition
of the structure; this allows us to apply structural measures
based on tree decompositions to reason about the intermediate degree of plans. 
As a first step towards our main result, we show a restricted version of the main theorem to SPJ plans, that is, an algorithm that, given an SPJ plan $p$ and a set $\Sigma$ of unary keys, outputs a semantically equivalent SPJ plan $p'$ (over all databases satisfying $\Sigma$) with the lowest possible intermediate degree.
While the way that this algorithm computes the minimized plan $p'$ from a given plan $p$ is, on a high level, quite natural --- essentially, $p'$ is based on the core of the chase of a relational structure naturally derived from $p$ --- the proof that $p'$ is minimal in the desired sense is non-trivial, since we have to reason about arbitrary plans that are semantically equivalent to $p'$. 
Intuitively, the core is a minimized version of a relational structure and the chase is a well-known procedure that enforces constraints on a relational structure.

As defined, the intermediate degree is a measure of query plans. It follows from our definitions that any \emph{algorithm for evaluating a query plan $q$} that materializes all intermediate relations must require time at least $\Omega(M_\reld^d)$, where $d$ is the intermediate degree of $q$.
That is, the polynomial time
 dependence on $M_\reld$ must be of degree at least $d$.
%whose polynomial degree is bounded below by the intermediate degree; that is, the intermediate degree lower bounds any such algorithm's polynomial degree.  
Having observed this lower bound, we are naturally led to the question of whether or not this bound is tight.
We present results indicating that it is, which renders the intermediate degree as an algorithmically meaningful measure.  
When given a plan~$p$, our minimization algorithm can always output a plan $p'$ with the properties described above, and with the additional syntactic property that $p'$ is \emph{well-behaved}; 
as we show, well-behavedness of a plan $q$ implies that there is an evaluation algorithm for evaluating $q$ on a database $\reld$
that always runs within time $M_\reld^d$ times a multiplicative overhead, where $d$ is the intermediate degree of $q$.  So, a well-behaved plan can be algorithmically evaluated in time close to $O(M_\reld^d)$, complementing the time lower bound of $\Omega(M_\reld^d)$ discussed above. 
%(We leave an analysis of the exact nature of the multiplicative overhead as a question for future work; our presentation and analysis of our evaluation algorithm are relatively elementary, and we wish to emphasize the point that the time bound that we obtain is obtained in an elementary fashion.)
We leave an analysis of the exact nature of the multiplicative overhead as a question for future work; we wish to emphasize here that the analysis of our evaluation algorithm and the resulting time bound are obtained in an elementary fashion.

{\bf Related Work.}
Our work is inspired by and is closely related to
the work of Atserias, Grohe, and Marx~\cite{GroheMarx14-fractional-edge-covers,AtseriasGroheMarx13-size-bounds}, as well as the follow-up work thereof by
Gottlob, Lee, Valiant, and Valiant~\cite{GottlobLeeValiant12-size-tw-bounds}.
A main result shown by~\cite{AtseriasGroheMarx13-size-bounds} yields a general upper bound---the \emph{AGM bound}---on the
output size of any pure join plan $q$, namely, a bound of $M_{\reld}^{\rho^*(q)}$, where $\rho^*(q)$ is the so-called fractional edge cover number of the plan $q$.
They also supply an argument that this upper bound is tight. 
Here, by a pure join plan, we mean 
a plan that consists only of joins and basic relations of a database.

Let us define the \emph{output degree} of a query to be the minimum degree of a polynomial bounding the query's output relation, as a function of 
the maximum relation size $M_{\reld}$ of a database $\reld$.
Note that while the output degree is 
clearly invariant under passing from a query to another query
that is semantically equivalent to the first, the intermediate degree
does not enjoy this invariance.
The AGM bound along with the argument that it is tight establish $\rho^*(q)$ as the output degree
of a pure join plan $q$.
The article~\cite{AtseriasGroheMarx13-size-bounds}
also exhibits, for each such pure join plan, a semantically equivalent
join-project plan whose intermediate relations all have size 
less than (or equal to) the output size, implying 
a characterization of the 
minimum 
intermediate degree (over semantically equivalent plans) as the fractional edge cover number---for pure join plans.
We view the article~\cite{AtseriasGroheMarx13-size-bounds} as begging the study of more
general classes of plans: 
%beyond pure join plans: 
they show explicitly
the necessity of using projections in the plans witnessing
that the intermediate relation sizes can be made to be at most
the output size!

Extending the AGM bound, the work~\cite{GottlobLeeValiant12-size-tw-bounds} characterizes 
the output degree of any SPJ plan and considering databases where unary key constraints
can be expressed; this is done by identifying a measure called 
the \emph{color number}.
%, and applying it to the chase of the plan.
Note that this measure generalizes the fractional edge cover number
(as discussed in Section 4 of \cite{GottlobLeeValiant12-size-tw-bounds}).
Our work builds on this measure and the accompanying understanding of unary key constraints. In particular,
in the present work, we naturally extend the notion of color number
to define a measure of tree decompositions.
We remark that the results of~\cite{GottlobLeeValiant12-size-tw-bounds}
are phrased in terms of
conjunctive queries, which are semantically equivalent to SPJ plans.

%By a pure join plan, we mean 
%a plan that joins together basic relations of a database;
%this output size is described using the minimum size of a 
%fractional edge cover of a hypergraph naturally derived from the 
%relations joined.  

An important and influential line of research emerging after
the work of Atserias, Grohe, and Marx presented new join algorithms,
so-called worst-case optimal join algorithms,
with running times shown to be bounded above by
the AGM bound---apart from a multiplicative overhead
(see for example~\cite{NgoRR13-skew-strikes-back,NgoPRR18-jacm-wco-join,Ngo18-wco-join-survey} and the discussion therein).
Building on this work, further research presented algorithms, with
running time analyses, for more general forms of queries, including queries with
aggregates~\cite{JoglekarPR16-ajar,KhamisNR16-faq}.  
The focus of these further works are on presenting algorithms and \emph{upper bounding} their running time in terms of structural measures, generally without considering how to minimize queries up to semantic equivalence; the present article establishes a complementary result by showing, for each query plan~$p$, % (of the specified form), 
a \emph{tight lower bound} on the running time of a restricted class of procedures, namely, procedures that can be described by plans that are semantically equivalent to the given plan~$p$.  

{\bf Perspective.} One measure of query plans that has been studied heavily is \emph{width}, which by now has a mature underlying theory (refer for example to~\cite{KolaitisVardi00-containment,Grohe07-otherside,Chen14-existentialpositive,BovaChen14-width-ep,Otto17,BovaChen19-howmany,ChenMengel24}); while this measure is often defined on logical formulas (and is sometimes formulated as the \emph{number of variables}), for a query plan the \emph{width} amounts to the maximum arity over all subplans.  The present article is a contribution to a general research direction whereby we hope to identify further meaningful and informative measures of query plans, and study how to minimize various classes of queries with respect to introduced measures.

%\mtodo{Extend related work with other width notions.}

%{\bf Contribution and Outline.} Our main contribution is an algorithm that finds, given a SPJU plan $p$ and a set $\Sigma$ of unary keys, a $\Sigma$-semantically equivalent SPJU plan $p'$ which has the lowest possible intermediate degree among all $\Sigma$-semantically equivalent SPJU plans. Towards this result, we show that an analogous theorem holds for SPJ plans, which bears the bulk of the technical development of this work. 

{\bf Outline.} 
In Section~\ref{sec:prelim}, we present the necessary concepts used throughout the paper. Section~\ref{sect:main-theorem-statement} defines the central notion of intermediate degree, gives an intuition for the main theorem as well as discusses some algorithmic aspects of the plan output by the algorithm underlying the main theorem. In Section~\ref{sect:spj} we show how we can prove the main theorem restricted to SPJ plans before generalizing our result to SPJU plans in Section~\ref{sect:spju}.

\section{Preliminaries}\label{sec:prelim}

We use $\Q$ to denote the rational numbers, $\Q^+$ to denote the non-negative rational numbers, and $\N$ to denote the natural numbers, which we understand to include $0$. For a natural number $k \in \N$, we use $[k]$ to denote the first $k$ positive natural numbers, that is, the set $\{ i \in \N ~|~ 1 \leq i \leq k \}$.
For a map $f: A \to B$ and some $S \subseteq A$, $f \upharpoonright S$ denotes the restriction of $f$ to $S$.
%When $\tup{a} = (a_1, \ldots, a_k)$ is a tuple, we use the notation $\{ \tup{a} \}$ to denote the set $\{ a_1, \ldots, a_k \}$; when $(a_1, \ldots, a_k)$ is a tuple over a set $A$ and $f: A \to B$ is a map, we generally use $f(a_1, \ldots, a_k)$ to denote the entry-wise action of $f$ on the tuple, that is, to denote $(f(a_1), \ldots, f(a_k))$.
For a tuple $\tup{a} = (a_1, \ldots, a_k)$ over a set $A$ and a map $f: A \to B$, by extension of notation, we write $f(\tup{a})$ for the entry-wise action of $f$ on the tuple, i.e., to denote $(f(a_1), \ldots, f(a_k))$; we further write $\{ \tup{a} \}$ for the set $\{ a_1, \ldots, a_k \}$.
When $A$ is a set, we use $A^*$ to denote the set of all finite-length tuples over $A$.
When~$\approx$ is a binary relation on a set $A$, we define $\approx^*$ as the reflexive-symmetric-transitive closure of~$\approx$, that is, $\approx^*$ is the smallest equivalence relation on $A$ that contains $\approx$ as a subset.

\subsection{Structures}

%We define a \emph{signature} to be a finite set of relation symbols, where each symbol $R$ has associated with it an \emph{arity}, which is a natural number denoted by $\ar(R)$.
A \emph{signature} $\sigma$ is a finite set of relation symbols, each with an associated arity. We denote by $R/n$ that symbol $R$ has arity $n\in\N$; we may also write $\ar(R)$ for $n$.
A \emph{structure} $\rela$ over a signature $\sigma$ consists of a finite set called the \emph{universe} or {\em domain} of the structure, denoted $\adom{\rela}$, and, for each symbol $R \in \sigma$, a relation $R^{\rela}$ over $\adom{\rela}$ whose arity is that of $R$.
%We generally use the boldface letters $\rela, \relb, \ldots$ to denote structures, and the corresponding letters $A, B, \ldots$ to denote their respective universes.
An \emph{isolated element} of a structure $\rela$ is a universe element $a \in \adom{\rela}$ such that $a$ does not appear in any tuple of any relation of~$\rela$.
An \emph{open structure} $(\rela,\tup{a})$ over signature $\sigma$ is a pair consisting of a structure $\rela$ over signature $\sigma$ with no isolated elements and a tuple $\tup{a}$ over the universe $\adom{\rela}$ of~$\rela$.
%The following measure of a structure will be crucial in this article: 
For a structure~$\rela$ over signature $\sigma$, let $M_{\rela}$ be the maximum size over all relations of~$\rela$, that is, $\max_{R \in \sigma} |R^{\rela}|$.

Two (open) structures are \emph{similar} if they are defined over the same signature.
When~$\rela$ and~$\relb$ are similar structures over the signature~$\sigma$, we define their union, denoted by $\rela \cup \relb$, as the structure with universe $\adom{\rela} \cup \adom{\relb}$ and where, for each symbol $R \in \sigma$, it holds that $R^{\rela \cup \relb} = R^{\rela} \cup R^{\relb}$.
For two similar structures~$\rela$ and~$\relb$ over a signature~$\sigma$, a \emph{homomorphism} from $\rela$ to $\relb$ is a map $h: \adom{\rela} \to \adom{\relb}$ such that for each symbol $R \in \sigma$, it holds that $\tup{a} \in R^{\rela}$ implies $h(\tup{a}) \in R^{\relb}$.
For similar open structures $(\rela,\tup{a})$ and $(\relb,\tup{b})$ with $\tup{a}$ and $\tup{b}$ having the same arity, a homomorphism $h$ from $(\rela,\tup{a})$ to $(\relb,\tup{b})$ is a homomorphism from $\rela$ to $\relb$ such that $h(\tup{a})=\tup{b}$.
Two similar structures $\rela$ and $\relb$ are \emph{homomorphically equivalent} if there exist a homomorphism from $\rela$ to $\relb$ and a homomorphism from $\relb$ to $\rela$.
An \emph{isomorphism} from $\rela$ to $\relb$ is a bijection $h: \adom{\rela} \to \adom{\relb}$ such that $h$ is a homomorphism from $\rela$ to $\relb$, and $h^{-1}$ is a homomorphism from $\relb$ to $\rela$.
%
%There is a known correspondence between Boolean SPJ queries and structures~\cite{ChandraMerlin77-optimal}; by a \emph{Boolean} SPJ query, we mean one that evaluates to true or false (or alternatively, a relation of arity $0$) on a given structure $\reld$. This correspondence yields that each such query can be translated to a structure $\rela$ such that, for an arbitrary structure $\reld$, the query is true on $\reld$ if and only if there is a homomorphism from $\rela$ to $\reld$.
%We prove a version of this correspondence in Theorem~\ref{thm:p-rep}.
%However, the version that we establish concerns SPJ plans that can output relations of any finite arity, and not just relations of arity $0$. In order to formulate our version, we define the following notion of \emph{open structure}.
%An \emph{open structure} is a pair $(\rela,\tup{a})$ consisting of a structure $\rela$ and a finite-length tuple $\tup{a}$ over the universe $A$ of~$\rela$; such an open structure is said to be over signature $\sigma$ when $\rela$ is over signature $\sigma$.
%
%The tuple $\tup{a}$ corresponds to variables that are \emph{free}, or, in database terminology, \emph{non-projected}.
Let $(\rela,\tup{a})$ be an open structure and let $\reld$ be a structure that is similar to $\rela$; we use $\homs(\rela,\tup{a},\reld)$ to denote the relation 
\[\{ h(\tup{a}) ~|~ \mbox{$h$ is a homomorphism from $\rela$ to $\reld$} \}.\]
We conceive of this relation as the result of evaluating the open structure on the structure $\reld$.
Similarly, when $\rela$ and $\reld$ are similar structures and $S \subseteq \adom{\rela}$, we use $\homs(\rela,S,\reld)$ to denote the set of maps
\[\{ h \upharpoonright S ~|~ \mbox{$h$ is a homomorphism from $\rela$ to $\reld$} \}.\]
%We sometimes refer to the mappings in these sets $\homs(\cdot,\cdot,\cdot)$ as \emph{answers}.

\begin{comment}
\begin{example}
\label{ex:small}
As a small example, let $\sigma$ be the signature containing a single relation $E$ of arity $2$, and let $\rela_0$ be the structure with universe $\adom{\rela_0} = \{ u, v_1, v_2 \}$ and with relation $E^{\rela_0} = \{ (u,v_1), (u,v_2) \}$. Consider the open structure $(\rela_0, (u))$; when $\reld$ is a directed graph, we have that $\homs(\rela_0, (u), \reld)$ is equal to the arity $1$ relation containing each vertex (of $\reld$) having an outgoing edge. Next, consider the open structure $(\rela_0, (v_1,v_2))$; when $\reld$ is a directed graph, we have that $\homs(\rela_0, (v_1,v_2), \reld)$ is equal to the arity $2$ relation containing each pair of vertices (of $\reld$) that each receive an incoming edge from a common vertex. This relation is symmetric, and when $x$ is a vertex of $\reld$, this relation contains the pair $(x,x)$ if and only if $x$ receives an incoming edge (from some vertex).
\end{example}%    
\end{comment}
The next fact follows from the classic work of Chandra and Merlin~\cite{ChandraMerlin77-optimal}.

\begin{proposition}[\cite{ChandraMerlin77-optimal}] \label{prop:cm}
Let $(\rela,\tup{a})$ and $(\rela',\tup{a'})$ be open structures over some signature $\sigma$. There exists a homomorphism from $(\rela,\tup{a})$ to $(\rela',\tup{a'})$ if and only if for each structure $\reld$ over $\sigma$ we have that $\homs(\rela',\tup{a'},\reld) \subseteq \homs(\rela,\tup{a},\reld)$.
\end{proposition}

%For the sake of completeness, we provide a proof in the appendix.

For similar structures $\rela$ and $\relb$, we say that $\rela$ is a \emph{substructure} of $\relb$ if $\adom{\rela} \subseteq \adom{\relb}$ and for each symbol $R \in \sigma$, it holds that $R^{\rela} \subseteq R^{\relb}$.
When $\relb$ is a structure and $S \subseteq \adom{\relb}$, the \emph{induced substructure} $\relb[S]$ is defined as the structure with universe~$S$ and where, for each symbol $R$, it holds that $R^{\relb[S]} = \{ (b_1, \ldots, b_k) \in R^{\relb} ~|~ \{ b_1, \ldots, b_k \} \subseteq S \}$.
When $\rela$ is a substructure of $\relb$, a \emph{retraction} from $\relb$ to $\rela$ is a homomorphism from $\relb$ to $\rela$ that acts as the identity on $\adom{\rela}$.
%fixes each element $a \in A$.
When there exists a retraction from $\relb$ to $\rela$, we say that $\relb$ \emph{retracts} to $\rela$.  A \emph{core of a structure} $\relb$ is a substructure $\relc$ of $\relb$ such that $\relb$ retracts to~$\relc$, but $\relc$ does not retract to any proper substructure of $\relc$, that is, any substructure of $\relc$ whose universe is a proper subset of $\relc$'s universe. 
A structure is a \emph{core} if it is a core of itself. 
We provide an example illustrating the notions of (open) structure, the relation $\homs(\rela,\tup{a},\reld)$, substructure, and core for the interested reader in the appendix.

\begin{comment}
\begin{example}
Consider the structure $\rela_0$ of Example~\ref{ex:small}. For each $i = 1, 2$, let $\rela_i$ be the structure with universe $\adom{\rela_i} = \{ u, v_i \}$ and with relation $E^{\rela_i} = \{ (u,v_i) \}$. Clearly, each of $\rela_1,\rela_2$ is a substructure of $\rela_0$. In addition, we have that the map fixing $u$ and $v_1$ and sending $v_2$ to $v_1$ is a retraction from $\rela_0$ to $\rela_1$; dually, the map fixing $u$ and $v_2$ and sending $v_1$ to $v_2$ is a retraction from $\rela_0$ to $\rela_2$.  Each of $\rela_1,\rela_2$ is a core of $\rela_0$ (and of itself): to argue this, we need to argue that (say) $\rela_1$ has no retraction to a proper substructure $\rela'$ (of $\rela_1$). (With respect to $\rela_0$, the structures $\rela_1$, $\rela_2$ are symmetric to each other.)
We argue this as follows: if we had such a proper substructure, it would have a size $1$ universe,
and the retraction would map each element in $\adom{\rela_1}$ to the single element $x$ in that size $1$ universe. But then this proper substructure $\rela'$ would have $E^{\rela'} = \{ (x, x) \}$, which could not be a subset of~$E^{\rela_1}$, contradicting that $\rela'$ is a substructure of $\rela_1$.
\end{example}
\end{comment}

The following facts are known and straightforward to verify.

\begin{proposition}
\label{prop:core}
The following facts hold.
\begin{enumerate}

\item For each structure $\relb$, there exists a core of $\relb$.
\item A core $\relc$ of a structure $\relb$ is homomorphically equivalent
to $\relb$.
\item Any two cores of a structure are isomorphic.  (We thus, by a mild abuse of terminology, speak of \emph{the} core of a structure.)
\item When $\relc$ is a core, and $\relb$ is a structure
that is homomorphically equivalent to $\relc$, 
there exists a core $\relc'$ of $\relb$ that is isomorphic to $\relc$.

\end{enumerate}
\end{proposition}

An open structure is not a structure per se. However, it will be useful to translate an open structure to a structure.
%; under this translation, we will be able to speak of homomorphisms between open structures, a notion which will be used crucially, via an upcoming fact (Proposition~\ref{prop:cm}).
%The translation we use is the following.
We reserve a set $\{R_k\}_{k\in\N}$ of relation names for the purpose of translating open structures to structures, where each relation $R_k$ has arity $k$ and is distinct from the relations mentioned in any non-augmented signature. 
When $(\rela,(a_1,\ldots,a_k))$ is an open structure over the signature $\sigma$, we define its \emph{augmented structure}, denoted by $\aug(\rela,(a_1,\ldots,a_k))$, as the structure $\rela_+$ over signature $\sigma \cup \{ R_k \}$ 
%, where $R_k$ is a fresh relation symbol with associated arity $k$. We then let ...
where $R^{\rela_+} = R^{\rela}$ for each $R \in \sigma$, and $R_k^{\rela_+} = \{ (a_1, \ldots, a_k) \}$. 
%Whenever we form the augmented structure of an open structure over a signature $\sigma$, we assume that $R_k$ is a fresh symbol not contained in $\sigma$.
In particular, this means that if open structures $(\rela,\tup{a})$ and $(\rela',\tup{a'})$ are similar and tuples $\tup{a}$ and $\tup{a'}$ have the same arity, then also the augmented structures $\aug(\rela,\tup{a})$ and $\aug(\rela',\tup{a'})$ are similar.

\OMIT{
\begin{example}
Regarding the structure introduced in Example~\ref{ex:small}, the augmented structure $\rela_+$  of the open structure $(\rela_0, (v_1,v_2))$ is over the signature $\{ E, R_2 \}$ where each of $E$ and $R_2$ have arity $2$. This augmented structure has relations $E^{\rela_+} = \{ (u,v_1), (u,v_2) \}$ and $R_2^{\rela_+} = \{ (v_1, v_2) \}$.
\end{example}
}

We associate each open structure with its augmented structure,
so for example we regard an open structure $(\relc,\tup{c})$ as
a core of an open structure $(\relb,\tup{b})$ when 
$\aug(\relc,\tup{c})$ is a core of $\aug(\relb,\tup{b})$.
Note that a homomorphism from 
an open structure $(\rela,\tup{a})$ to an open structure $(\relb,\tup{b})$,
as we defined it above, is a homomorphism from $\aug(\rela,\tup{a})$ to $\aug(\relb,\tup{b})$, and vice-versa.

%We associate each open structure with its augmented structure; for example, when $(\rela,\tup{a})$ and $(\relb,\tup{b})$ are open structures over the same signature where $\tup{a}$ and $\tup{b}$ have the same length, by a homomorphism from $(\rela,\tup{a})$ to $(\relb,\tup{b})$, we mean a homomorphism from $\aug(\rela,\tup{a})$ to $\aug(\relb,\tup{b})$.

\subsection{Query plans}
\label{subsect:query-plans}

For completeness and clarity, we precisely define the types of plans to be studied.
For each $m \geq 0$, we define an $m$-suitable identification as an expression $j = k$ where $j, k \in [m]$.

For a signature $\sigma$, a \emph{SPJU plan} over $\sigma$ is inductively defined as follows:

\begin{itemize}
\item (basic)
For each $R \in \sigma$, it holds that $R$ is a plan of arity $\ar(R)$.
\item (select)
When $p$ is a plan of arity $m$ and $\theta$ is a set of $m$-suitable identifications,  $\sigma_{\theta}(p)$ is a plan of arity $m$.
\item (project)
When $p$ is a plan of arity $m$ and $j_1, \ldots, j_n$ is a sequence of numbers from $[m]$, with $n \geq 0$, it holds that $\pi_{j_1, \ldots, j_n}(p)$ is a plan of arity $n$.
\item (join)
When $p_1, \ldots, p_\ell$ are plans of arity $m_1, \ldots, m_\ell$, respectively, with $\ell \geq 1$, and $\theta$ is a set of $(m_1 + \cdots + m_\ell)$-suitable identifications, $\Join_{\theta}(p_1, \ldots, p_\ell)$ is a plan of arity $(m_1 + \cdots + m_\ell)$.
\item (union)
When $p_1, p_2$ are plans of arity $m$, it holds that $p_1 \cup p_2$ is a plan of arity $m$.
\end{itemize}

In this paper, we will sometimes simply use the term \emph{plan} 
to refer to an SPJU plan.
We define an \emph{SPJ plan} to be a SPJU plan that does not make use of union.

Suppose that $Q$ is a relation of arity $m$. For any set $\theta$ of $m$-suitable identifications, we define 
%\begin{center}
$\sigma_{\theta}(Q) = \{ (d_1, \ldots, d_m) \in Q ~|~ \mbox{for each $(j=k)\in\theta$, $d_j=d_k$ holds} \}.$
%\end{center}
For any sequence $j_1, \ldots, j_n \in [m]$, we define 
%\[
$\pi_{j_1, \ldots, j_n}(Q) = \{ (d_{j_1}, \ldots, d_{j_n}) ~|~ (d_1, \ldots, d_m) \in Q \}.$
%\]

Let $p$ be a SPJU plan over a signature $\sigma$, let $\reld$ be a structure over $\sigma$, and let $m$ be the arity of~$p$. We define the \emph{output of $p$ over $\reld$}, denoted by $\out(p,\reld)$, to be the arity $m$ relation over~$D$ defined as follows.
Note that, as is usual in this context, we understand that there is a unique tuple of arity $0$, called the \emph{empty tuple}.
\begin{itemize}
\item
When $p$ has the form $R$, we define $\out(p,\reld) = R^{\reld}.$

\item
When $p$ has the form $\sigma_{\theta}(p')$, we define $\out(p,\reld) = \sigma_{\theta}( \out(p',\reld)).$

\item
When $p$ has the form $\pi_{j_1, \ldots, j_n}(p')$, we define $\out(p,\reld) = \pi_{j_1, \ldots, j_n}( \out(p',\reld) ).$

\item
When $p$ has the form $\Join_{\theta}(p_1, \ldots, p_\ell)$, we define $\out(p,\reld) = \sigma_{\theta}( \out(p_1,\reld) \times \cdots \times \out(p_\ell,\reld)).$

\item When $p$ has the form $p_1 \cup p_2$, we define $\out(p,\reld) = \out(p_1,\reld) \cup \out(p_2,\reld).$
\end{itemize}

Observe that selection is a special case of join, in particular, it corresponds to a $1$-way join.
We say that two plans $p, p'$ over the same signature $\sigma$ are \emph{semantically equivalent} when, for every structure $\reld$ over $\sigma$, it holds that $\out(p,\reld) = \out(p',\reld)$. 
Further, for a set $\Sigma$ of keys over $\sigma$, we say that two plans $p, p'$ over $\sigma$ are \emph{$\Sigma$-semantically equivalent} when, for every structure $\reld$ over $\sigma$ that satisfies $\Sigma$, it holds that $\out(p,\reld) = \out(p',\reld)$.
A \emph{subplan} of a plan is defined in the usual manner. 
Formally, we define a set $\subplans(p)$, for each plan $p$, inductively as follows:
\begin{itemize}
\item
When $p$ has the form $R$, we define $\subplans(p) = \{ p \}$.

\item
When $p$ has the form $\sigma_{\theta}(p')$ or $\pi_{j_1, \ldots, j_n}(p')$, we define $\subplans(p) = \subplans(p') \cup \{ p \}.$

\item
When $p$ has the form $\Join_{\theta}(p_1, \ldots, p_\ell)$, we define $\subplans(p) = (\bigcup_{i = 1}^\ell \subplans(p_i)) \cup \{ p \}.$

\item When $p$ has the form $p_1 \cup p_2$, we define $\subplans(p) = \subplans(p_1) \cup \subplans(p_2) \cup \{ p \}.$

\end{itemize}
We then say that $q$ is a \emph{subplan} of $p$ when $q \in \subplans(p)$.

\subsection{Hypergraphs and tree decompositions}

A \emph{hypergraph} is a pair $(V,E)$ where $V$ is a set, called the \emph{vertex set}, and $E \subseteq 2^V$ is a set called the \emph{edge set}.  
A \emph{graph} is a hypergraph where each edge has size $2$. 
When $H$ is a hypergraph, we use $V(H)$ and $E(H)$ to denote its vertex set and edge set, respectively.
An \emph{isolated vertex} of a hypergraph $H$ is a vertex not occurring in any edge, that is, an element of $V(H) \setminus \bigcup_{e \in E(H)} e$.
When $H$ is a hypergraph and $S \subseteq V(H)$, we use $H[S]$ to denote the induced subhypergraph of $H$ on $S$, that is, the hypergraph with vertex set $S$ and with edge set $\{ e \cap S ~|~ e \in E(H) \}$.

When $\rela$ is a structure over signature $\sigma$ we use $\H(\rela)$ to denote the hypergraph with vertex set $\adom{\rela}$ and edge set $\{ \{ \tup{a} \} ~|~ R \in \sigma, \tup{a} \in R^{\rela} \}$.
When $(\rela,\tup{a})$ is an open structure, we use $\H(\rela,\tup{a})$ to denote the hypergraph
with vertex set $\adom{\rela}$ and with edge set $E(\H(\rela)) \cup \{ \{ \tup{a} \} \}$.

A \emph{tree decomposition} of a hypergraph $H$ is a pair $(T,\chi)$, where $T$ is a tree and $\chi$ is a labeling function $V(T) \ra 2^{V(H)}$, i.e., it assigns each vertex of $T$ a subset of vertices of $H$, such that 
\begin{enumerate}
    \item (vertex coverage) for each $v \in V(H)$, there exists $t \in V(T)$ such that $v \in \chi(t)$,
    \item (edge coverage) for each $e \in E(H)$, there exists $t \in V(T)$ such that $e \subseteq \chi(t)$, and
    \item (connectivity) for each $v \in V(H)$, the set $\{ t\in V(T) ~|~ v \in \chi(t) \}$ induces a connected subtree of~$T$.
\end{enumerate}
Note that, if $H$ has no isolated vertices, condition $(1)$ is implied by condition $(2)$.
A tree decomposition of an open structure $(\rela,\tup{a})$, is a tree decomposition of its hypergraph $\H(\rela,\tup{a})$.

\subsection{Key constraints}

A \emph{key constraint} $\kappa$ over signature $\sigma$ is an expression $\textrm{key}(R) = K$, where $R/n \in\sigma$ and $K\subseteq [n]$. For an $n$-tuple $\tup{a} = (a_1,\ldots,a_n)$ and some $K = \{i_1,\ldots,i_m\}\subseteq [n]$, for some $m\in [n]$, we write $\tup{a}[K]$ for the tuple $(a_{i_1},\ldots,a_{i_m})$. A structure $\rela$ over signature $\sigma$ satisfies $\kappa$ if, for every two tuples $\tup{a},\tup{b} \in R^{\rela}$, $\tup{a}[K] = \tup{b}[K]$ implies $\tup{a} = \tup{b}$. If, for some key $\kappa$ of the form $\textrm{key}(R) = K$, the set $K$ is a singleton, we say that $\kappa$ is a unary key. We say a structure satisfies a set of keys $\Sigma$, if it satisfies every key $\kappa\in\Sigma$. 

The \emph{chase} is a well-known tool for enforcing key dependencies (see, e.g.,~\cite{MaierMendelzohnSagiv79-testing-data-dependencies, Deutsch06-query-reformulation}).
For an open structure $(\rela,\tup{a})$ over signature $\sigma$ and a set $\Sigma$ of keys over $\sigma$ a \emph{chase step} on $(\rela,\tup{a})$ w.r.t.~$\Sigma$ is a substitution of the form $x \mapsto y$, which is applied globally on $\rela$ and $\tup{a}$, where there are two tuples $\tup{b}, \tup{b'}\in R^{\rela}$ and a key $\textrm{key}(R) = K$ for some $R/n\in\sigma$, such that $\tup{b}[K] = \tup{b'}[K]$ but there exists some $i\in [n]\setminus K$ such that $\tup{b}[i]=x$ and $\tup{b'}[i]=y$. The \emph{chase} of $(\rela,\tup{a})$ w.r.t.~$\Sigma$, denoted $\textrm{Chase}_\Sigma(\rela,\tup{a})$, is the result of repeated applications of chase steps until a fixpoint is reached; this fixpoint is known to be unique up to isomorphism. The chase of a structure $\rela$ over signature $\sigma$ w.r.t.~$\Sigma$, denoted $\textrm{Chase}_\Sigma(\rela)$, is defined analogously in the obvious way.

We next define a measure for hypergraphs based on a notion, called color number, defined on conjunctive queries and introduced in~\cite{GottlobLeeValiant12-size-tw-bounds}.  
For a structure $\rela$ over signature $\sigma$ and a set $\Sigma$ of keys over $\sigma$, a \emph{valid coloring} of $\rela$ w.r.t. $\Sigma$ is a mapping $\mathsf{Col}_{\Sigma}^{\rela}: \adom{\rela} \ra 2^{\{1,\ldots,k\}}$, for some $k>0$, which assigns to each element of the universe of $\rela$ a subset of $k$ colors such that the following hold:
\begin{enumerate}
    \item For each key $\textrm{key}(R) = \{i_1,\ldots,i_m\}\in\Sigma$, i.e., $R/n\in\sigma$ and $\{i_1,\ldots,i_m\}\subseteq [n]$, we have that for every $\tup{a}\in R^{\rela}$ and every $i\in [n]\setminus\{i_1,\ldots,i_m\}$ it holds that $\mathsf{Col}_{\Sigma}^{\rela}(\tup{a}[i]) \subseteq \bigcup_{j\in[m]}\mathsf{Col}_{\Sigma}^{\rela}(\tup{a}[i_j])$.
    \item There exists an element $x\in \adom{\rela}$ such that $\mathsf{Col}_{\Sigma}^{\rela}(x) \neq \emptyset$.
\end{enumerate}
For a structure $\rela$ over signature $\sigma$, a set $S\subseteq\adom{\rela}$, and a set $\Sigma$ of keys over $\sigma$ the \emph{color number} of $S$ w.r.t.~$\rela$ and $\Sigma$, denoted $C_{\Sigma}^{\rela}(S)$, is defined as the maximal ratio of colors assigned to the elements in $S$ over the maximum number of colors assigned to the elements of some tuple in $\rela$ which is attainable by any valid coloring (for any $k$) of $\rela$ w.r.t.~$\Sigma$. Formally, where $\max_{\mathsf{Col}_{\Sigma}^{\rela}}$ ranges over all such valid colorings:
\[C_\Sigma^{\rela}(S) = \max_{\mathsf{Col}_{\Sigma}^{\rela}}\frac{|\bigcup_{x\in S}\mathsf{Col}_{\Sigma}^{\rela}(x)|}{\max_{\tup{b}\in\rela}|\bigcup_{x\in\tup{b}}\mathsf{Col}_{\Sigma}^{\rela}(x)|}.\]
A crucial observation is that for a coloring that maximises the color number of $S$, we have that $x\notin S$ implies $\mathsf{Col}_{\Sigma}^{\rela}(x) = \emptyset$. Note that, for a structure $\rela$ over signature $\sigma$, the color number $C_{\Sigma}^{\rela}(S)$, as defined here, is equal to the color number, as defined in~\cite{GottlobLeeValiant12-size-tw-bounds}, of the conjunctive query whose atoms are those in 
$\{ R(v_1,\ldots,v_k) ~|~ R \in \sigma, (v_1,\ldots,v_k) \in R \}$,
and whose output variables are $S$.
We refer the reader to~\cite{GottlobLeeValiant12-size-tw-bounds} for more information on the color number, but note the following fact here.
\begin{proposition} (refer to ~\cite[Section 3]{GottlobLeeValiant12-size-tw-bounds}) %the last paragraph of
\label{prop:alpha}
Let $\relb$ be a structure, and let $S$ be a subset of $\adom{\relb}$.
The value $C^\relb_{\emptyset}(S)$ is the fractional edge cover number of the induced hypergraph of $\H(\relb)$ on $S$. 
\end{proposition}
%
%By the induced hypergraph of a hypergraph $G = (V,F)$ on a subset $S \subseteq V$, we mean the hypergraph $(S, \{ e \cap S ~|~ e \in F \})$.
A \emph{fractional edge cover} of a hypergraph $G = (V,F)$ is a mapping $x: F \to [0,\infty)$ such that for each $v \in V$, it holds that  $\sum_{e \in E, v \in e} x(e) \geq 1$; the sum $\sum_{e \in E} x(e)$ is the \emph{weight} of $x$.  The \emph{fractional edge cover number} of $G$ is the minimum weight over all fractional edge covers of $G$; it follows from the theory of linear programming that a minimum exists~\cite{GroheMarx14-fractional-edge-covers}.

%; here, $\sigma$ denotes the signature
%of $\rela$.

%Note that $C_{\Sigma}^{\rela}$ is essentially a mapping from the subsets of the vertices of the hypergraph $\H(\rela)$ to $\Q^+$, since the vertices of the hypergraph are by definition elements of the universe of $\rela$. 

Let us view  $C_{\Sigma}^{\rela}$ as a mapping from the power set of 
$\adom{\rela}$ to $\Q^+$.
We define the following width notions. For an open structure $(\rela,\tup{a})$ over signature $\sigma$, a set $\Sigma$ of keys over $\sigma$, 
%the hypergraph $\H(\rela) = (V(H),E(H))$, 
and a tree decomposition $F=(T, \chi)$ of $\H(\rela,\tup{a})$, 
%and the mapping $C_{\Sigma}^{\rela} : 2^{V(H)} \ra \Q^{+}$, 
let $C_{\Sigma}^{\rela}\wid(F) = \max \{ C_{\Sigma}^{\rela}(\chi(t)) ~|~ t \in V(T) \}$. Further, let $C_{\Sigma}^{\rela}\wid(\rela,\tup{a}) = \min_F C_{\Sigma}^{\rela}\wid(F)$, where $F$ ranges over all tree decompositions of $\H(\rela,\tup{a})$.
%Note that the order of elements in $\tup{a}$ does not impact the color number of an open structure $(\rela, \tup{a})$.
%, and thus, for the purpose of computing the color number for a subset $S$ of the universe of $\rela$, we can consider a canonical tuple $\tup{S}$ over the elements in $S$. 
%For an open structure $(\rela, \tup{a})$ over signature $\sigma$, potentially given with an associated set of keys $\Sigma$ over $\sigma$, we can thus define the $C_{\Sigma}^{\rela}\wid$ of $(\rela, \tup{a})$ as above.
%consider the hypergraph $\H(\rela)$ and a tree decomposition $F=(T = (V^T,E^T), \chi)$ of $\H(\rela)$. The \emph{color-width} of $F$ is defined as $\mathsf{color}\wid(F) = $

\OMIT{
Let $\gamma: 2^W \to \Q^+$ be a set function on a finite set $W$. Following Adler~\cite{Adler06-thesis},
we define the following width notions.
\begin{itemize}
\item When $F = (T,(B_t))$ is a tree decomposition where each bag is a subset of $W$, we define 
\[\gamma\wid(F) = \max \{ \gamma(B_t) ~|~ t \in V(T) \}.\]
(Note that the possible values of $\gamma\wid(F)$ are those in the set $\{ \gamma(U) ~|~ U \subseteq W \}$.)
\item When $H$ is a hypergraph with $V(H) \subseteq W$, we define
\[\gamma\wid(H) = \min_F \gamma\wid(F),\]
where the minimum is over all tree decompositions $F$ of $H$.
\end{itemize}

\new{
Following Adler~\cite{Adler08-tw-and-functional-dependencies}, we define the following width notions. For a hypergraph $H = (V(H),E(H))$, a tree decomposition $F=(T = (V(T),E(T)), \chi)$ of $H$ and a mapping $\gamma : 2^{V(H)} \ra \Q^{+}$, let $\gamma\wid(F) = \max \{ \gamma(\chi(t)) ~|~ t \in V(T) \}$. Further, let $\gamma\wid(H) = \min_F \gamma\wid(F)$, where the minimum is over all tree decompositions $F$ of $H$.
}

\new{
Following Adler~\cite{Adler08-tw-and-functional-dependencies}, we define the following width notions. For an open structure $(\rela,\tup{a})$, its hypergraph $\H(\rela,\tup{a}) = (V(H),E(H))$, a tree decomposition $F=(T = (V(T),E(T)), \chi)$ of $\H(\rela,\tup{a})$, and a mapping $\gamma : 2^{V(H)} \ra \Q^{+}$, let $\gamma\wid(F) = \max \{ \gamma(\chi(t)) ~|~ t \in V(T) \}$. Further, let $\gamma\wid((\rela,\tup{a})) = \min_F \gamma\wid(F)$, where the minimum is over all tree decompositions $F$ of $\H(\rela,\tup{a})$.
}
}

\section{Degree Notions and Main Theorem Statement}
\label{sect:main-theorem-statement}

%\% Repeat some motivation from introduction re output and intermediate degree?
%In this section, we present our main theorem, and make some preparatory technical observations.  
%We begin by presenting the notion of \emph{intermediate degree}, along with a closely related notion called \emph{output degree}---first on an intuitive level, and then formally.

\subsection{Degree notions}

A key measure in the context of query optimization is the size of the output of a query. Here, we are interested in an asymptotic measure: the output degree of a query plan $p$, which, informally, is the smallest possible degree of a polynomial bounding the output size of $p$ on any structure~$\reld$ satisfying given constraints, as a function of the maximum relation size $M_\reld$; the formal definition is the following.

%We say a query plan $p$ has output degree $d$ if there exists a degree $d$ polynomial bounding the output size of $p$ on any structure~$\reld$, as a function of the maximum relation size $M_\reld$---and there is no lower degree polynomial providing such a bound. 

\begin{definition}[Output Degree]
    Let $\Sigma$ be a set of keys over signature~$\sigma$. Let $p$ be a plan over $\sigma$, and let $d \in \Q$. The plan~$p$ has \emph{output degree $\leq d$} if there exists a function $h: \N \to \N$, with $h(M) \in O(M^d)$, such that for every structure $\reld$ satisfying $\Sigma$, it holds that $|\out(p,\reld)| \leq h(M_\reld)$. 
    We say that $p$ has \emph{output degree $d$} if it has output degree $\leq d$, but for all $\epsilon > 0$, it does not have output degree $\leq (d - \epsilon)$. 
\end{definition}

Instead of just focusing on some plan as a whole, we are interested in bounding also the size of the output of every subplan of the original plan, or in other words, bounding the size of intermediate results obtained during the execution of the plan. 
To do so, we say a plan $p$ has intermediate degree $d$ if there exists a degree $d$ polynomial bounding the sizes of all intermediate relations computed during the execution of $p$ on any structure $\reld$ satisfying given constraints, as a function of $M_\reld$---and again, there is no lower degree polynomial providing such a bound.  

\begin{definition}[Intermediate Degree]
    Let $\Sigma$ be a set of keys over signature $\sigma$. Let $p$ be a plan over $\sigma$, and let $d \in \Q$. The plan $p$ has \emph{intermediate degree $\leq d$} if there exists a function $h: \N \to \N$, with $h(M) \in O(M^d)$, such that for every structure $\reld$ satisfying $\Sigma$, and for every subplan $q$ of $p$, it holds that $|\out(q,\reld)| \leq h(M_\reld)$.
    We say that $p$ has \emph{intermediate degree $d$} if it has intermediate degree $\leq d$, but for all $\epsilon > 0$, it does not have intermediate degree $\leq (d - \epsilon)$.
\end{definition}

Note that the above definitions of output degree and intermediate degree are with respect to the keys $\Sigma$, but we omit $\Sigma$ when it is clear from context for the sake of readability.
%We say that a plan $p'$ has the \emph{best possible intermediate degree} if no other semantically equivalent plan achieves a strictly lower intermediate degree, that is, $p'$ has the lowest intermediate degree among all semantically equivalent plans.
Let $\Sigma$ be a set of keys.
We say that a plan $p$ has the \emph{SPJ-best possible intermediate degree} if there exists $d \in \Q$ such that $p$ has intermediate degree $d$, and for every SPJ plan $p'$ that is $\Sigma$-semantically equivalent to $p$, there is no value $\epsilon > 0$ such that the plan $p'$ has intermediate degree $\leq (d-\epsilon)$.
In other words, no $\Sigma$-semantically equivalent SPJ plan achieves an intermediate degree that is strictly lower than that of $p$. 
Analogously, we say that a plan $p$ has the \emph{SPJU-best possible intermediate degree} if there exists $d \in \Q$ such that $p$ has intermediate degree $d$, and for every SPJU plan $p'$ that is $\Sigma$-semantically equivalent to $p$, there is no value $\epsilon > 0$ such that the plan $p'$ 
 has intermediate degree $\leq (d-\epsilon)$.

In the appendix, we give a number of examples of plans along with computations of the output degree and the best possible intermediate degree.

\subsection{Main theorem statement}

We are now ready to formulate our main theorem. 

\begin{theorem}[Main theorem] \label{thm:main-spju}
%Let $\Sigma$ be a set of unary keys over signature $\sigma$. 
%There exists an exponential-time algorithm that, given an SPJU plan $p$ over $\sigma$, outputs a SPJU plan $p'$ satisfying the following properties:
There exists an exponential-time algorithm that, given a signature~$\sigma$, a set $\Sigma$ of unary keys over $\sigma$, and an SPJU plan $p$ over $\sigma$, outputs a SPJU plan $p'$ satisfying the following properties:
\begin{itemize}
\item $p'$ is $\Sigma$-semantically equivalent to $p$,
\item $p'$ has the SPJU-best possible intermediate degree.
\end{itemize}
\end{theorem}

By an exponential-time algorithm, we mean an algorithm that runs within time $O(2^{Q(n)})$, where $n$ is the input size and $Q$ is a polynomial. 

The principal argumentation needed to establish our main theorem is located in the case where this theorem is specialized to SPJ queries.
In this case, which we present next, we can describe the intermediate degree of the output plan quite directly, but require the following notion to do so.
A \emph{$p$-representation} of a plan $p$ is an open structure $(\rela,\tup{a})$ where, for any structure $\reld$, we have $\out(p,\reld) = \homs(\rela,\tup{a},\reld)$. A formal treatment of $p$-representation and related notions utilized to show our results will follow in Subsection~\ref{ssect:query-structures}.

\begin{theorem} \label{thm:main-spj}
%Let $\Sigma$ be a set of unary keys over signature $\sigma$. 
%There exists an exponential-time algorithm that, given an SPJ plan $p$ over $\sigma$, outputs a SPJ plan $p'$ satisfying the following properties:
There exists an exponential-time algorithm that, given a signature~$\sigma$, a set $\Sigma$ of unary keys over $\sigma$, and an SPJ plan $p$ over $\sigma$, outputs a SPJ plan $p'$ satisfying the following properties:
\begin{itemize}
\item $p'$ is $\Sigma$-semantically equivalent to $p$,
\item $p'$ has the SPJ-best possible intermediate degree.
\end{itemize}
Moreover, suppose that $p$ is an SPJ plan over $\sigma$ and 
that $(\rela,\tup{a})$ is a $p$-representation; let $(\relc,\tup{c})$ be a core of $\textrm{Chase}_\Sigma(\rela,\tup{a})$, and let $d = C_{\Sigma}^{\relc}\wid(\relc,\tup{c})$.
Then, when the algorithm is given $p$ as input, the output plan $p'$ has intermediate degree $d$, and also has size polynomial in the size of $p$.
\end{theorem}

Some remarks are in order. First, note that in Theorem~\ref{thm:main-spj}, the intermediate degree of each plan~$p'$ can be described quite explicitly.
%Note also, that given a plan $p'$ output by the algorithm of Theorem~\ref{thm:main-spj}, its intermediate degree can be computed in polynomial time;
Note also that given an SPJ plan $p$, the intermediate degree $d$ of the plan $p'$ given by the algorithm of Theorem~\ref{thm:main-spj} is computable in exponential time from $p$; 
this follows directly from the exponential-time computability of $(\relc,\tup{c})$ from $p$ and the polynomial time computability of the color number, which follows from results in Sections~3 and 4 from~\cite{GottlobLeeValiant12-size-tw-bounds}.
In turn, given an SPJU plan $p$, the intermediate degree of a plan $p'$ output by Theorem~\ref{thm:main-spju}'s algorithm is also computable in exponential time from $p$, which follows from the above fact and the construction of $p'$ in the proof of Theorem~\ref{thm:main-spju}.
We also wish to remark that, while our main theorem only considers constraints which are unary keys, this limitation arises from the literature on output degree bounds;  extending the literature's bounds is likely to be challenging, for this would constitute significant progress towards a long-standing open problem in information theory~\cite[Section 1]{GottlobLeeValiant12-size-tw-bounds}.  %\htodo{can be give a pointer into this reference?}

\subsection{Algorithmic aspects}

%\htodo{this section has to be written properly.}

%Towards establishing our positive results, that is, presenting SPJ plans satisfying certain upper bounds, we rely on the notion of well-behaved plans.
%Let us make a couple of remarks on our formalization of plans.
%First, observe that selection is a special case of join, in particular, it corresponds to a $1$-way join.
%Second, note that we allow multiway joins (following~\cite{AtseriasGroheMarx13-size-bounds}). These are used to establish our positive results, where we present plans satisfying certain upper bounds. As regards such results, we view the sizes of relations defined by these multiway joins as faithful to the time complexity of computing these relations in accordance with the plans, since each plan synthesized to establish positive results will be \emph{well-behaved} in the following sense.
We here discuss algorithms for evaluating \emph{well-behaved} plans in time that is close to a plan's intermediate degree; this issue was motivated in the introduction.
We define a plan $p$ to be \emph{well-behaved} 
if for each subplan $q$ that is a join $\Join_\theta(p_1,\ldots,p_\ell)$ of plans $p_1,\ldots,p_\ell$ having arities $m_1,\ldots,m_\ell$,
%letting all indices range over the values in $[m_1 + \cdots + m_\ell]$,
the following is satisfied, when we let
 $s$ denote the sum $m_1 + \cdots + m_\ell$:
there exists an index $k \in [s]$ such that for each index $i \in [s] \setminus ([m_1] \cup \{ k \})$, there exists an index $j \in [m_1] \cup \{ k \}$ such that $(i = j) \in \theta$.
In words, 
each join adds at most one column of new information 
to the plan $p_1$: letting $J$ denote the set containing
(1) the indices in $[m_1]$, which correspond to $p_1$,
and (2) another chosen index $k$; it must be that every 
index $i \notin J$ is equated with an index in $J$.
% and another index $k$, every other index $i$ is equated with one of 
This property allows for the efficient computation of such joins and, thus, well-behaved plans; we will make this formal in Proposition~\ref{prop:well-behaved}. 
First, we point out that our minimization algorithms can be implemented in such a way that they always output well-behaved plans.
\begin{theorem}
The algorithms of Theorems \ref{thm:main-spju} and \ref{thm:main-spj} 
can be implemented in such a way so that each plan $p'$ output by them is well-behaved.
\end{theorem}

This easily follows from the fact that the algorithm of Theorem~\ref{thm:plan-synthesis} always outputs well-behaved SPJ plans and that the SPJU plan output by the algorithm of Theorem~\ref{thm:main-spju} is simply the union of well-behaved SPJ plans, which is also a well-behaved plan.

\begin{comment}
\begin{proof}
In the case of SPJ queries, this theorem follows from the observation that every plan output by the algorithm of Theorem~\ref{thm:main-spj} comes from 
the algorithm of Theorem~\ref{thm:plan-synthesis}, which always outputs well-behaved plans. In the case of SPJU queries, this theorem follows from the observation that every plan output by the algorithm of Theorem~\ref{thm:main-spju} is the union of plans output by the algorithm of Theorem~\ref{thm:main-spj}.
\end{proof}    
\end{comment}

%Note that the algorithm of Theorem~\ref{thm:main-spj} is \emph{meta} as an algorithm:

%The algorithm provided by Theorem~\ref{thm:main-spj} passes from a plan $p$ to a second plan $p'$ with the best possible intermediate degree. A natural question that arises is how much time it takes to evaluate $p'$ on a given structure $\reld$.

The following proposition shows that, given a well-behaved plan, apart from a multiplicative overhead depending on the plan, the time needed to algorithmically evaluate the plan on a database~$\reld$ is at most $|\adom{\reld}| * M_\reld^{d}$, where $d$ is the intermediate degree of the plan---so in essence, the intermediate degree of the plan dictates the polynomial degree of evaluation (as a function of $M_\reld$). 
To argue this, we in essence use the argumentation for \cite[Theorem 6]{AtseriasGroheMarx13-size-bounds}, and we assume a standard random access model of computation, with a uniform cost measure.

\begin{proposition} \label{prop:well-behaved}
There exists an algorithm that, given a well-behaved SPJU plan $p$ of intermediate degree $\leq d$ and a structure $\reld$, evaluates $p$ on $\reld$ in time $O(|p|^2 * |\adom{\reld}| * M_\reld^{d})$.
%, where $R$ is a computable mapping sending each SPJ plan $q$ to a natural number.
\end{proposition}

To argue this proposition, we consider an algorithm that evaluates the plan $p$ inductively. It suffices to show that each subplan $q\in\subplans(p)$ can be computed in time $O(|p| * |\adom{\reld}| * M_\reld^{d})$. 
By the definition of intermediate degree, each subplan $q$ has at most $O(M_\reld^d)$ answers (tuples in $\out(p,\reld)$). When $q$ is a basic plan, the required bound holds immediately. When $q$ is a select or project subplan, i.e., of the form $\sigma_{\theta}(q')$ or $\pi_{j_1,\ldots,j_n}(q')$, respectively, the answers of $q$ can be naturally derived from the answers of the immediate subplan $q'$, within the given time bound.  When $q$ is a union subplan, i.e., of the form $q_1\cup q_2$, it is immediate that the answers of $q$ can be derived by combining the answers of $q_1$ and $q_2$ and eliminating duplicates, which is feasible in the given time bound.
The key point is that, when $q$ is a join subplan, by well-behavedness, it is of the form $\Join_\theta(q_1,\ldots,q_\ell)$ and $\out(q,\reld)$ can be computed by maintaining a relation that is a subset of $\out(q_1,\reld)\times \adom{\reld}$. When evaluating the multiway join, we can thus evaluate the sequence of binary joins $\Join_{\theta_{\ell-1}}(\ldots(\Join_{\theta_2}(\Join_{\theta_1}(q_1,q_2),q_3),\ldots),q_\ell)$, where $\theta_1,\ldots,\theta_{\ell-1}$ are the corresponding equalities from $\theta$ for the join partners in question. Each binary join can be evaluated in time $O(|\adom{\reld}| * M_\reld^{d})$ by~\cite{FlumFrickGrohe02-query}, and in total, there are at most $|p|$ many joins, giving us the desired runtime of $O(|p| * |\adom{\reld}| * M_\reld^{d})$ to evaluate the subplan $q$.

%Again by intermediate degree, the result of each subquery, denoted $|q_1|$ and $|q_2|$, is of size at most $O(M_\reld^d)$. 
%Crucially, the well-behavedness of $p$ guarantees that the output of the join is of arity at most one more than the bigger arity of either subplan. Thus, the size of the output of $q$, denoted $|q|$, is at most $O(M_\reld * M_\reld^d)$. Together with the fact from~\cite{FlumFrickGrohe02-query} that a join $q = \Join(q_1,q_2)$ can be evaluated in time $O(|q_1| + |q_2| + |q|)$, and the fact that the check for $\theta$ can be evaluated in the required time as before, we have that $q$ can be computed in time $O(M_\reld^{d+1})$, as required.

%of plans $p_1$ and $p_2$ having arity $m_1$ and $m_2$, respectively, such that there exists an index $k\in\{m_{m_1}+1,\ldots,m_1+m_2\}$ such that for all $i\in \{m_{m_1}+1,\ldots,m_1+m_2\}\setminus \{k\}$ there exists $j\in [m_1]$ where $(i=j)\in\theta$.  In words, each join is a binary join of two plans $p_1$ and $p_2$, and the second plan $p_2$ is only allowed to add up to one column of new information.
%joins together $p_1, \ldots, p_\ell$, by well-behavedness,
%there is an index $i$ 

\section{SPJ Plans}
\label{sect:spj}

%\mtodo{Shorten subsections, move detailed proofs to appendix}

This section is dedicated to proving Theorem~\ref{thm:main-spj}, which, apart from being of interest in its own right, serves as a building block towards proving our main result in the next section.

Proving Theorem~\ref{thm:main-spj} essentially requires us to show two statements: (i) we can synthesize a plan, which is $\Sigma$-equivalent to the original one and has intermediate degree $\leq d$ and (ii) any plan which is $\Sigma$-equivalent to the original one has a subplan with output degree at least $d$. We will formalize and prove (i) and (ii) in Theorems~\ref{thm:plan-synthesis} and~\ref{thm:lowerbound-intermediate-sjp-color}, respectively, in Subsection~\ref{ssect:bounds}.

To show the theorem, we make two more basic observations on the introduced degree notions.
%In the remainder of this section, we make some basic observations on the introduced degree notions that will aid us in the sequel. 
We begin by presenting a definition that will aid us in arguing that a plan has a certain output degree. We say that $p$ has \emph{output degree $\geq d$} if there exists a function $g: \N \to \N$ with $g(M) \in \Omega(M^d)$ and there exists an infinite sequence $(\reld_n)_{n\geq1}$ of structures with $M_{\reld_n}$ unbounded such that (for all $n \geq 1$) $|\out(p,\reld_n)| \geq g(M_{\reld_n})$. We then have the following lemma.

\begin{lemma} \label{lem:degree-lower-bound}
Let $d \geq 0$. If a plan has output degree $\geq d$, then for any $\epsilon > 0$, it does not have output degree $\leq (d-\epsilon)$.
\end{lemma}

Next, we characterize the intermediate degree of a plan in terms of the output degrees of subplans, in the following sense.

\begin{lemma} \label{lem:intermediate-output}
Let $d \geq 0$. A plan $p$ has intermediate degree $\leq d$ if and only if each subplan $q$ of $p$ has output degree $\leq d$.
\end{lemma}

%The proofs of both lemmas can be found in the appendix. 
Theorem~\ref{thm:main-spj} can then be argued as follows: Given a plan $p$, we first compute the open structure $(\rela,\tup{a})$, which can be done in polynomial time by Theorem~\ref{thm:p-rep}. Then, the algorithm of Theorem~\ref{thm:main-spj} computes $\textrm{Chase}_\Sigma(\rela,\tup{a})$, again in polynomial time.  
It then computes the core of this open structure; this can be performed in exponential time. Let $(\relc,\tup{c})$ be this core.
Note that the size of $(\relc,\tup{c})$ is bounded above by the size of $\textrm{Chase}_\Sigma(\rela,\tup{a})$, since the core of a structure is a substructure of the structure. Let $d = C_{\Sigma}^{\relc}\wid(\relc,\tup{c})$. 
From $(\relc,\tup{c})$, we can then synthesize the desired plan $p'$ of intermediate degree $\leq d$ by statement (i), formalized in Theorem~\ref{thm:plan-synthesis}, from above.
%applying the algorithm of Theorem~\ref{thm:plan-synthesis}.
Consider any plan $p''$ that is $\Sigma$-semantically equivalent to $p'$. By statement (ii), formalized in Theorem~\ref{thm:lowerbound-intermediate-sjp-color},
%Theorem~\ref{thm:lowerbound-intermediate-sjp-color}, 
$p''$ has a subplan with output degree $\geq d$. For any $\epsilon > 0$, by Lemma~\ref{lem:degree-lower-bound}, this subplan does not have output degree $\leq ( d - \epsilon)$, and so by Lemma~\ref{lem:intermediate-output}, $p''$ does not have intermediate degree $\leq ( d - \epsilon)$. 
Hence, we have computed a $\Sigma$-equivalent plan $p'$ of polynomial size and of intermediate degree $d$, which is the best possible intermediate degree.

%It allows us to upper bound the number of homomorphisms from a structure $\rela$ to a structure $\reld$, with respect to a subset $S$ of the universe $A$ of $\rela$.

In the following, we will formalize the connection from plans to open structures and auxiliary notions, which are needed to prove the theorems formalizing statements (i) and (ii) in Subsection~\ref{ssect:bounds}.

%%%%%%%%%%%%%%%%%%%%%%%%%%%%%%%%%%%%%%%%%%%%%%%%%%%

\subsection{Relating Queries to Structures}
\label{ssect:query-structures}
There is a known correspondence between Boolean SPJ plans and structures~\cite{ChandraMerlin77-optimal}; by a \emph{Boolean} SPJ plan, we mean one that evaluates to true or false (or alternatively, a relation of arity $0$) on a given structure $\reld$. This correspondence yields that each such plan can be translated to a structure $\rela$ such that, for an arbitrary structure $\reld$, the plan is true on $\reld$ if and only if there is a homomorphism from $\rela$ to $\reld$.
We prove a version of this correspondence in Theorem~\ref{thm:p-rep} for SPJ plans that can output relations of any finite arity, and not just relations of arity $0$. In order to formulate our version, we utilize the notion of open structure.

%The output degree of an SPJ plan can be characterized by a known result.
\begin{definition}[$p$-Representation]
    Let $p$ be an SPJ plan over signature $\sigma$. A \emph{$p$-representation} is an open structure $(\rela,\tup{a})$ over $\sigma$, where for each structure $\reld$ over $\sigma$, it holds that $\out(p,\reld) = \homs(\rela,\tup{a},\reld)$.
\end{definition}
%For an SPJ query plan $p$ over $\sigma$, we define a \emph{$p$-representation} as an open structure $(\rela,\tup{a})$ over $\sigma$ where the structure $\rela$ does not have any isolated elements, and, for each structure $\reld$ over $\sigma$, it holds that $\out(p,\reld) = \homs(\rela,\tup{a},\reld)$. 
%
%In this translation from a query plan $p$ to an open structure $(\rela,\tup{a})$, $\rela$ comes from the known translation from Boolean queries to structures, while $\tup{a}$ essentially collects the free variables of $p$. %; we will make this precise in a moment.

Apart from showing that any SPJ plan $p$ can be translated to a $p$-representation, we also want to provide a tree decomposition of this open structure with certain additional properties, which we call $p$-decomposition.
%
%We begin by showing that any SPJ plan can be translated to an open structure along with a tree decomposition of that open structure's hypergraph.  
This allows us to relate plans to structural, decomposition-based measures.
While the translation that we present is, on a high level, along the lines of known translations from SPJ plans to conjunctive queries, we emphasize that our translation also provides a tree decomposition that will allow us to analyze the intermediate degree. In particular, for a given plan, each subplan will correspond to a vertex of the tree decomposition whose associated set of elements corresponds to the free variables of the subplan.  

We say a tree decomposition $(T,\chi)$ is \emph{rooted} if there is a distinguished node $v_0\in V(T)$, called the root of $T$. The formal definition of $p$-decomposition follows.

%Let $p$ be an SPJ query plan over $\sigma$.
%We show how to translate $p$ to a $p$-representation, 
%an open structure called
%a \emph{$p$-representation}, 
%along with a form of tree decomposition of this $p$-representation, called a \emph{$p$-decomposition}.  
%The precise definitions follow.
\begin{definition}[$p$-Decomposition]
    Let $p$ be an SPJ plan over $\sigma$ and $(\rela,\tup{a})$ a $p$-representation. A \emph{$p$-decomposition} of $(\rela,\tup{a})$ is a rooted tree decomposition $E = (T, \chi)$ of $(\rela,\tup{a})$ such that there exists
    \begin{itemize}
    %\item a distinguished node $v_0\in V(T)$, called the root of $T$,
    \item a bijection $\alpha: \subplans(p) \to V(T)$, where $\alpha(p)$ is the root of $T$, and 
    \item a map $\beta: \subplans(p) \to \adom{\rela}^*$, where $\beta(p)$ is equal to $\tup{a}$, 
    \end{itemize}
    where for each $q \in \subplans(p)$, we have that $\{ \beta(q) \} = \chi(\alpha(q))$.
    %and $q$ satisfies the containment property. 
\end{definition}

In the definition of $p$-decomposition, each subplan $q$ of $p$ corresponds to a vertex of the tree~$T$, via $\alpha$, and for each such subplan $q$, it holds that $\beta(q)$ is a tuple that lists, in some order (and possibly with repetition), the elements of the set to which $\chi$ maps q's tree vertex.
%the elements of the set assigned via $\chi$ of the tree vertex corresponding to $q$, namely, the tree vertex $\alpha(q)$.

Let $p$ be an SPJ plan over signature $\sigma$, $(\rela,\tup{a})$ be a $p$-representation, and $E = (T,\chi)$ be a rooted $p$-decomposition $E = (T,\chi)$ of $(\rela,\tup{a})$ with associated mappings $\alpha$ and $\beta$. For a vertex $v\in T$, let $S_v$ be the union of $\chi(u)$ over all descendants $u$ of $v$ (including $v$ itself) and let $\rela_{\leq v, E}$ denote the induced substructure $\rela[S_v]$.
%
%For an SPJ plan $p$ over signature $\sigma$, a $p$-representation $(\rela,\tup{a})$, and a rooted $p$-decomposition $E = (T,\chi)$ of $(\rela,\tup{a})$, 
We say a subplan $q \in \subplans(p)$ satisfies the \emph{containment property} w.r.t.~a $p$-decomposition if, for each structure $\reld$ over~$\sigma$, the relationship $\out(q,\reld) \supseteq \homs(\rela_{\leq \alpha(q),E}, \beta(q), \reld)$ holds. Further, we say a $p$-decomposition satisfies the containment property, if every subplan $q \in \subplans(p)$ satisfies the containment property w.r.t.~the $p$-decomposition.
%
%Here, when $\rela$ is a structure and $E = (T,\chi)$ is a rooted tree decomposition whose bags are subsets of the universe $\adom{\rela}$, and $s$ is a vertex of $T$, 
%we use $\rela_{\leq s, E}$ to denote the induced substructure $\rela[S]$, 
%where~$S$ is the union of $\chi(u)$ over all descendants $u$ of $s$ (including $s$ itself).
%
%We remark that
%, in the definition of $p$-decomposition, 
%requiring the containment property to hold on $p$ itself is a matter of formulation; in the case where we consider $q = p$, the containment property holds as a consequence of the definition of $p$-representation.
The containment property will allow us to lower bound the sizes of intermediate relations of $p$ using a $p$-representation along with a $p$-decomposition satisfying the containment property.

\begin{example}
    Consider the signature $\sigma$ containing the binary relations $R/2$ and $S/2$ and the SPJ plan $p = \Join_{(1=3)}(R, \pi_{1}(S))$ over $\sigma$. Then it can be readily seen, with structure $\rela = \{R(a,b),S(a,d)\}$ and tuple $\tup{a} = (a,b,a)$, that $(\rela,\tup{a})$ is a $p$-representation. Consider the following trees:
    
    \begin{figure}[h]
	   \centering
	   \includegraphics[width=.5\textwidth]{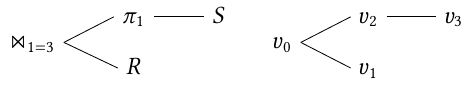}
	   %\caption{The plan $p$ and tree $T$}
	   %\label{fig:decomp-example}
    \end{figure}

    %\vspace{-0.1em}
    On the left, we have the plan $p$ represented as a tree (each node represents the subplan of its subtree), while on the right we have some tree $T$. We provide mappings $\alpha$, $\beta$, and $\chi$ such that $(T,\chi)$ is a $p$-decomposition via mappings $\alpha$ and $\beta$. First, $\alpha$ is the obvious bijection from the subplans of $p$ to $V(T)$ with $\alpha(p)=v_0$, $\alpha(R)=v_1$, $\alpha(\pi_1(S))=v_2$, and $\alpha(S)=v_3$. Further, we define $\beta$ as $\beta(p)=\tup{a}$, $\beta(R)=(a,b)$, $\beta(\pi_1(S))=(a)$, and $\beta(S)=(a,d)$. Lastly, define $\chi(\alpha(q))$ as $\{ \beta(q) \}$ for all subplans $q$ of $p$. It can be readily verified that $(T,\chi)$ is indeed a $p$-decomposition, and that it satisfies the containment property.
    %Then show the $p$-representation and $p$-decomposition with containment property.
\end{example}

We show with the following theorem, that a $p$-representation and a $p$-decomposition satisfying the containment property %with the above properties 
always exist and are efficiently computable, via an inductive argument over the structure of $p$ given in the appendix.

%If in addition 
%\[
%\out(p,\reld) = \homs(\rela\langle \approx \rangle, \beta(p), \reld).
%\]
%holds
%for each structure $\reld$ over $\sigma$,
%we say that the $p$-decomposition is a
%\emph{strong $p$-decomposition}.

\begin{theorem}
\label{thm:p-rep}
There exists a polynomial-time algorithm that, given an SPJ plan $p$ over signature $\sigma$, outputs a $p$-representation over $\sigma$ along with a $p$-decomposition of this $p$-representation which satisfies the containment property.
\end{theorem}

\subsection{Intermediate Size Bounds}
\label{ssect:bounds}

We devote this subsection to formalizing and establishing the statements used to prove Theorem~\ref{thm:main-spj}.
We start with statement (i), i.e. we show that we can synthesize a plan which is $\Sigma$-equivalent to the one given originally and which has the desired intermediate degree.

\begin{theorem}[Plan synthesis] 
\label{thm:plan-synthesis}
There exists an exponential-time algorithm that, given a set $\Sigma$ of unary keys over a signature $\sigma$ and an open structure $(\relc,\tup{c})$ over $\sigma$ with $\textrm{Chase}_\Sigma(\relc,\tup{c}) = (\relc,\tup{c})$ and which is a $p$-representation for some plan $p$, outputs a well-behaved plan $p'$ where
%(a) $(\relc,\tup{c})$ is a $p'$-representation, 
(a) $p'$ is $\Sigma$-semantically equivalent to $p$, and 
(b) $p'$ has intermediate degree $\leq C_{\Sigma}^{\relc}\wid(\relc,\tup{c})$.
\end{theorem}

To show the above theorem, we combine proof techniques from~\cite{AtseriasGroheMarx13-size-bounds} and~\cite{GottlobLeeValiant12-size-tw-bounds}. As a first step, following arguments from Section 4 of~\cite{GottlobLeeValiant12-size-tw-bounds}, we transform the open structure $(\relc,\tup{c})$ together with $\Sigma$ to an open structure $(\relc',\tup{c})$ (over a modified signature $\sigma'$) which has no constraints posed over it. This transformation can be achieved via the application of a well-behaved SPJ plan of intermediate degree 1, and where every new relation is of a size that is bounded from above by an original relation. Crucially, for any structure $\reld$ satisfying $\Sigma$, we have that $(\relc,\tup{c})$ returns the same answers on $\reld$ as $(\relc',\tup{c})$ on an analogous transformation of $\reld$. Moreover, we have that $C_{\emptyset}^{\relc'}\wid(\relc',\tup{c}) \leq C_{\Sigma}^{\relc}\wid(\relc,\tup{c})$ and it thus suffices to show the result for the empty set of unary keys $\Sigma' = \emptyset$ along with $(\relc',\tup{c})$, in place of $\Sigma$ with $(\relc,\tup{c})$. Finally, we compute a tree decomposition $F$ of $(\relc',\tup{c})$ where $C_{\Sigma'}^{\relc'}\wid(F)$ is equal to $C_{\Sigma'}^{\relc'}\wid(\relc',\tup{c})$, which can be done in exponential time, and then compute a SPJ plan $p'$ following the tree composition using standard techniques from~\cite{AtseriasGroheMarx13-size-bounds}. We provide a detailed proof in the appendix.

We continue to show that statement (ii) 
%of Section~\ref{sect:main-theorem-statement} 
also holds. To that end, let us first formalize the statement.
%
%Towards establishing Theorem~\ref{thm:sjp-intermediate}, we show the following result relating, on the one hand, the defined degree $e$, and on the other hand, the plans semantically equivalent to the plan $p$ of interest.
In essence, it provides a lower bound on the intermediate degree of any plan that is $\Sigma$-semantically equivalent to $p$.

\begin{theorem}
\label{thm:lowerbound-intermediate-sjp-color}
Let $\Sigma$ be a set of unary keys over signature $\sigma$. Let $p$ be an SPJ plan over $\sigma$ and $(\rela,\tup{a})$ be a $p$-representation. Let $\textrm{Chase}_\Sigma(\rela,\tup{a}) = (\relb,\tup{b})$, and let $(\relc,\tup{c})$ be a core of $(\relb,\tup{b})$. Let $d = C_{\Sigma}^{\relc}\wid(\relc,\tup{c})$.
For any SPJ plan $p^+$ that is $\Sigma$-semantically equivalent to $p$, it holds that $p^+$ has a subplan with output degree $\geq d$.
\end{theorem}

Before proceeding with the proof, we establish two preparatory lemmas that will aid us in showing the lower bound on the intermediate degree.
First, we show that for an open structure $(\relc,\tup{c})$ which is a core of the chase of a $p$-representation and which has width $d$, there is a vertex in any tree decomposition of $(\relc,\tup{c})$ witnessing that the intermediate results grow in accordance with a degree $d$ polynomial.
%every tree decomposition of a core of the chase of some open structure has a vertex 
%whose associated set of elements on which the number of answers can be made to grow in accordance with a degree $d$ polynomial; this is made formal as follows.

\begin{lemma}
\label{lemma:td-bag-color}
Let $\Sigma$ be a set of unary keys over signature $\sigma$. Let $p$ be an SPJ plan over $\sigma$ and $(\rela,\tup{a})$ be a $p$-representation. Let $\textrm{Chase}_\Sigma(\rela,\tup{a}) = (\relb,\tup{b})$, and let $(\relc,\tup{c})$ be a core of $(\relb,\tup{b})$. Let $d = C_{\Sigma}^{\relc}\wid(\relc,\tup{c})$.
For any tree decomposition $(T,\chi)$ of $(\relc,\tup{c})$, there exists a vertex $t_0$ of $T$, a constant $c' > 0$, and a sequence of structures $(\reld_n)_{n \geq 1}$ with $M_{\reld_n}$ unbounded where $|\homs(\relc, \chi(t_0), \reld_n)| \geq c' M_{\reld_n}^d$.
\end{lemma}

\begin{comment}

\begin{proof}
By the definition of $C_{\Sigma}^{\relc}\wid$ of the open structure $(\relc,\tup{c})$, there exists a vertex $t_0$ of the given tree decomposition such that $C_{\Sigma}^{\relc}\wid(\chi(t_0)) \geq d$.
Let $\tup{\chi(t_0)}$ be the tuple of elements in $\chi(t_0)$. Invoke Lemma~\ref{lemma:lower-color} on the structure $\relc$ and the open structure $(\relc,\tup{\chi(t_0)})$ to obtain a constant $c' > 0$ and structures $(\reld_n)_{n \geq 1}$ such that 
$|\homs(\relc,\chi(t_0),\reld_n)| \geq c' M_{\reld_n}^{C_{\Sigma}^{\relc}(\chi(t_0))}$.
%The present lemma follows.
\end{proof}    

\end{comment}

We next give a lemma whose assumptions are a relaxation of (that is, they are implied by) the condition that $\relc$ is a core of $\relb$.
It allows us to lower bound, for any subset $S \subseteq \adom{\relb}$, the size of $\homs(\relb,S,\cdot)$ by the size of $\homs(\relc,S \cap \adom{\relc}, \cdot)$.

\begin{lemma}
\label{lemma:substruct}
Suppose that $\relb, \relc$ are structures over signature $\sigma$ such that $\relc$ is a substructure of $\relb$, and there exists a retraction $h$ from $\relb$ to $\relc$.
Then, for any subset $S \subseteq \adom{\relb}$ and any structure $\reld$ over $\sigma$, it holds that $|\homs(\relb,S,\reld)| \geq |\homs(\relc, S \cap \adom{\relc}, \reld)|$.
\end{lemma}

\begin{comment}
 
\begin{proof}
We show how to pass from any map in $\homs(\relc, S~\cap~\adom{\relc}, \reld)$ to a map in $\homs(\relb,S,\reld)$, and then explain why the passage is injective, which suffices.
Suppose that $g \in \homs(\relc, S \cap \adom{\relc}, \reld)$ and $h$ is a retraction from $\relb$ to $\relc$. Then there exists an extension $g'$ of $g$ that is in $\homs(\relc, \adom{\relc}, \reld)$, from which it follows that the composition $g' \circ h$ is in $\homs(\relb, \adom{\relb}, \reld)$ (in writing this composition, $h$ is applied first).
It follows that the restriction $(g' \circ h) \upharpoonright S$ is in $\homs(\relb, S, \reld)$.  This passage is injective since if two mappings in $\homs(\relc, S \cap \adom{\relc}, \reld)$ differ on an element $c$ of $S \cap \adom{\relc}$, the two mappings that they are passed to will also differ on $c$, due to the assumption that $h$ acts as the identity on $C$.
\end{proof}
   
\end{comment}

We proceed to prove the final statement needed to show Theorem~\ref{thm:main-spj}.

\begin{proof}[Proof of Theorem~\ref{thm:lowerbound-intermediate-sjp-color}]
Fix a set $\Sigma$ of unary keys and a SPJ plan $p$ over $\Sigma$. Let $(\rela,\tup{a})$ be a $p$-representation, $(\relb,\tup{b}) = \textrm{Chase}_\Sigma(\rela,\tup{a})$, and $(\relc,\tup{c})$ be a core of $(\relb,\tup{b})$. For any structure $\reld$ satisfying $\Sigma$ we have that 
\[\out(p,\reld) = \homs(\rela,\tup{a},\reld) = \homs(\relb,\tup{b},\reld) = \homs(\relc,\tup{c},\reld),\]
where the first equality holds by definition of $p$-representation, the second equality holds by definition of the chase procedure and the third equality holds by Propositions~\ref{prop:core}.

Let $p^+$ be any SPJ plan that is $\Sigma$-semantically equivalent to $p$. Let $(\rela^+,\tup{a^+})$ be a $p^+$-representation and $(\relb^+,\tup{b^+}) = \textrm{Chase}_\Sigma(\rela^+,\tup{a^+})$. For any structure $\reld$ satisfying $\Sigma$ we have that 
\[\out(p,\reld) = \out(p^+,\reld) = \homs(\rela^+,\tup{a^+},\reld) = \homs(\relb^+,\tup{b^+},\reld),\]
where the first equality holds by $\Sigma$-semantic equivalence and the second and third equality hold by definition as above.

\OMIT{
We have, for each structure $\reld$, that
$\homs(\relc,\tup{c},\reld) = 
 \homs(\rela,\tup{a},\reld)$,
via Propositions~\ref{prop:core} and~\ref{prop:cm}.
Let $p^+$ be any SPJ plan that is semantically equivalent to $p$.
By Theorem~\ref{thm:p-rep},
there exists a $p^+$-representation 
$(\relb,\tup{b})$ 
as well as a $p^+$-decomposition $E^+ = (T,(B^+_t))$ thereof,
with associated mappings $\alpha,\beta$.
%Set $(\relb,\tup{b})$ to be equal to 
%$(\rela^+\langle \approx^+ \rangle, \tup{a^+}\langle \approx^+ \rangle)$.
We have, for each structure $\reld$, that
\[\homs(\relc,\tup{c},\reld) = \out(p,\reld) = \out(p^+,\reld) = \homs(\relb,\tup{b},\reld).\]  
}

%Thus by Proposition~\ref{prop:cm}, we have that $(\relc,\tup{c})$ and $(\relb,\tup{b})$ are homomorphically equivalent.
Since we have that $\homs(\relc,\tup{c},\reld) = \homs(\relb^+,\tup{b^+},\reld)$ for all structures $\reld$ satisfying $\Sigma$, and we know that both $\relc$ and $\relb^+$ are structures satisfying $\Sigma$, it holds specifically that $\tup{c} \in \homs(\relc,\tup{c},\relc) = \homs(\relb^+,\tup{b^+},\relc)$ and $\tup{b^+} \in \homs(\relb^+,\tup{b^+},\relb^+) = \homs(\relc,\tup{c},\relb^+)$. This immediately implies that there is a homomorphism from $(\relb^+,\tup{b^+})$ to $(\relc,\tup{c})$, and vice-versa, showing that the two open structures are homomorphically equivalent.
Since $(\relc,\tup{c})$ is a core, by Proposition~\ref{prop:core} we have that $\aug(\relc,\tup{c})$ is isomorphic to a substructure of $\aug(\relb^+,\tup{b^+})$. By renaming the elements of $\adom{\relb^+}$, we can achieve that $\aug(\relc,\tup{c})$ is a core of (and thus a substructure of) $\aug(\relb^+,\tup{b^+})$.  We then have that $\tup{b^+} = \tup{c}$.

By Theorem~\ref{thm:p-rep}, there exists a $p^+$-decomposition $E^+ = (T,\chi^+)$ of $(\rela^+,\tup{a^+})$ with associated mappings $\alpha,\beta$. Within $E^+$, we equate elements following the chase procedure on $(\rela,\tup{a})$ such that the resulting $E^+$ is a decomposition of $(\relb^+,\tup{b^+})$.
We define a tree decomposition of $(\relc,\tup{c})$ as follows. Take the pair $(T,\chi)$ where $T$ is as in $E^+$, and where, for each vertex $t$ of $T$, we define $\chi(t) = \chi^+(t) \cap \adom{\relc}$.  
It is straightforward to verify that $(T,\chi)$ is a tree decomposition of $(\relc,\tup{c})$.
%In the tree decomposition $(T,(B^+_t))$, the bag $B^+_r$ of the root vertex $r$
%was $\{ \tup{b} \}$ (by definition of $p^+$-decomposition),
%which is equal to $\{ \tup{c} \}$; we thus have that $B_r = \{ \tup{c} \}$.
By Lemma~\ref{lemma:td-bag-color}, there exists a vertex $t$ of $T$, a constant $c' > 0$, and structures $(\reld_n)_{n\geq 1}$ with $M_{\reld_n}$ unbounded, where $|\homs(\relc,\chi(t),\reld_n)| \geq c' M^d_{\reld_n}$.
There exists a subplan $q^+$ of $p^+$ such that $\alpha(q^+) = t$. We then have
%
\begin{comment}
\begin{align*}
|\out(q^+,\reld_n)| 
& \geq |\homs(\relb_{\leq t, E^+}, \chi^+(t), \reld_n)| \\
& \geq |\homs(\relb, \chi^+(t), \reld_n)| \\
& \geq |\homs(\relc, \chi(t), \reld_n)| \\
& \geq c' M^d_{\reld_n}. 
\end{align*}
\end{comment}
%
\begin{align*}
    |\out(q^+,\reld_n)| &\geq |\homs(\relb_{\leq t, E^+}, \chi^+(t), \reld_n)| \geq |\homs(\relb, \chi^+(t), \reld_n)|\\
    &\geq |\homs(\relc, \chi(t), \reld_n)| \geq c' M^d_{\reld_n}.
\end{align*}
The first inequality holds by the definition of $p^+$-decomposition; the second inequality holds since $\relb_{\leq t, E^+}$ is a substructure of $\relb$; the third inequality follows from Lemma~\ref{lemma:substruct}, whose assumptions hold because $\aug(\relc,\tup{c})$ is a core of $\aug(\relb,\tup{b})$; the last inequality follows from our choice of $(\reld_n)_{n\geq 1}$. It follows that the subplan $q^+$ has output degree $\geq d$.
\end{proof}

\section{SPJU Plans}
\label{sect:spju}

%\mtodo{Give proof of main theorem (Outline, details in appendix)}

Building on the established understanding of intermediate degree for SPJ plans from the previous section, we are now ready to prove Theorem~\ref{thm:main-spju}. 
%We prove Theorem~\ref{thm:main-spju}.
%To do so, we build on the established understanding of intermediate degree for SPJ plans.
%
It is well-known that each SPJU plan $p$ is $\Sigma$-semantically equivalent to some plan $p'$ which is a union $p_1 \cup \cdots \cup p_m$ of SPJ plans.
%(this follows, for example, from this article's Theorem~\ref{thm:sjpu});
Let us note that the output degree of an SPJU plan can then be characterized as follows: 
%it is readily verifiable that, 
for such a $\Sigma$-semantically equivalent plan $p'$, the output degree of $p$ is equal to the maximum output degree over the plans $p_1, \ldots, p_m$.
%
%In this section, we build on the established understanding of SPJ plans to prove the main theorem.  
%It is known that that any SPJU plan can be translated to the union of SPJ plans;
To argue about the intermediate degree of an SPJU plan, we first give a theorem showing that a natural such translation to a union of SPJ plans exists and can be computed in exponential time, which has the additional property that the relation computed by a subplan of one of the SPJ plans is a subset of some relation computed by a subplan of the original SPJU plan; this inclusion then facilitates proving lower bounds on the intermediate degree of the original SPJU plan.

\begin{theorem}
\label{thm:sjpu}
There exists an exponential-time algorithm that, given an SPJU plan $p$, outputs a sequence of SPJ plans $p_1, \ldots, p_m$, a sequence of maps $\alpha_1, \ldots, \alpha_m$, and a sequence of open structures $(\rela_1,\tup{a_1}), \ldots, (\rela_m,\tup{a_m})$ such that the following holds:
\begin{enumerate}

\item $p$ is semantically equivalent to $p_1 \cup \cdots \cup p_m$,

\item for each $i = 1, \ldots, m$, we have that $\alpha_i: \subplans(p_i) \to \subplans(p)$ is a map such that
(for each structure $\reld$) when $q \in \subplans(p_i)$,
the containment
$\out(q,\reld) \subseteq \out(\alpha_i(q),\reld)$ holds,

\item for each $i = 1, \ldots, m$, the open structure
$(\rela_i, \tup{a_i})$ is a $p_i$-representation,
and $i \neq j$ implies that there is no homomorphism 
from $(\rela_i,\tup{a_i})$ to  $(\rela_j,\tup{a_j})$.

\end{enumerate}
\end{theorem}

On a high level, Theorem~\ref{thm:main-spju} can then be proved as follows. For some SPJU plan $p$, the algorithm first computes this convenient set of SPJ plans $p_1, \ldots, p_m$ via Theorem~\ref{thm:sjpu}.
%invokes the algorithm of Theorem~\ref{thm:sjpu}
%on the given plan $p$;
Then, for each obtained plan $p_i$, it computes a core $(\relc_i,\tup{c_i})$ of the chase of a $p$-representation, and then applies Theorem~\ref{thm:main-spj} to obtain a value $d_i$, and a plan $p'_i$ such that $p_i$, $d_i$, and $p'_i$ satisfy the properties given in that theorem statement.
The algorithm outputs the value $d = \max_{i \in [m]} d_i$ and the plan $p' = p'_1 \cup \cdots \cup p'_m$. 
Since this algorithm outputs a union of well-behaved SPJ plans, also the resulting SPJU plan is well-behaved.
%We first give the theorem that this convenient set of SPJs exists and can be computed, followed by a proof that the resulting $p'$ has indeed the desired properties.
It is easy to see that $p'$ has intermediate degree $d$, it just remains to show that $d$ is indeed the SPJU-best possible intermediate degree. In essence, we prove that for any SPJU $q'$ that is semantically equivalent to $p'$, we have that from the sequence of SPJ plans $q'_1, \ldots, q'_n$ obtained by applying Theorem~\ref{thm:sjpu} to $q'$, there exists an index $\ell\in [n]$ such that $q'_\ell$ does not have (for any $\epsilon > 0$) intermediate degree $\leq (d-\epsilon)$ implying by condition (2) of Theorem~\ref{thm:sjpu} that $q'$ does not have (for any $\epsilon > 0$) intermediate degree $\leq (d-\epsilon)$. Full details can be found in the appendix.

%\end{comment}

\section{Conclusion}
\label{sect:conclusion}

The intermediate degree asymptotically characterizes the size of intermediate relations obtained during the execution of a query plan. For SPJU plans, we have shown how to optimize for this measure among all semantically equivalent SPJU plans with respect to a given set of unary key constraints. A natural way to extend our work is to consider either more expressive query languages, by adding more operators, or more expressive constraints, such as 
%arbitrary key constraints or 
arbitrary functional dependencies. 
%
%Summing up, we have shown how to compute an SPJ plan with the best possible intermediate degree and which is equivalent to a given one
%
%, given some SPJ plan $p$, how to compute a plan $p'$, which is equivalent to $p$ under given unary key constraints, and which has the best possible intermediate degree, i.e.
%
%Recap of work;
%Extension to SPJU in appendix; consider more expressive constraints such as keys or general functional dependencies; 

%%
%% Bibliography
%%

%% Please use bibtex, 

\bibliography{main-icdt26}

\appendix

\section{Examples}

\subsection{Example Illustrating Structure-related Notions (cf. Section~\ref{sec:prelim})}

Let $\sigma$ be the signature containing a single relation $E$ of arity $2$, and let $\rela_0$ be the structure with universe $\adom{\rela_0} = \{ u, v_1, v_2 \}$ and with relation $E^{\rela_0} = \{ (u,v_1), (u,v_2) \}$. Consider the open structure $(\rela_0, (u))$; when $\reld$ is a directed graph, we have that $\homs(\rela_0, (u), \reld)$ is equal to the arity $1$ relation containing each vertex (of $\reld$) having an outgoing edge. Next, consider the open structure $(\rela_0, (v_1,v_2))$; when $\reld$ is a directed graph, we have that $\homs(\rela_0, (v_1,v_2), \reld)$ is equal to the arity $2$ relation containing each pair of vertices (of $\reld$) that each receive an incoming edge from a common vertex. This relation is symmetric, and when $x$ is a vertex of $\reld$, this relation contains the pair $(x,x)$ if and only if $x$ receives an incoming edge (from some vertex).

Further, for each $i = 1, 2$, let $\rela_i$ be the structure with universe $\adom{\rela_i} = \{ u, v_i \}$ and with relation $E^{\rela_i} = \{ (u,v_i) \}$. Clearly, each of $\rela_1,\rela_2$ is a substructure of $\rela_0$. In addition, we have that the map fixing $u$ and $v_1$ and sending $v_2$ to $v_1$ is a retraction from $\rela_0$ to $\rela_1$; dually, the map fixing $u$ and $v_2$ and sending $v_1$ to $v_2$ is a retraction from $\rela_0$ to $\rela_2$.  Each of $\rela_1,\rela_2$ is a core of $\rela_0$ (and of itself): to argue this, we need to argue that (say) $\rela_1$ has no retraction to a proper substructure $\rela'$ (of $\rela_1$). (With respect to $\rela_0$, the structures $\rela_1$, $\rela_2$ are symmetric to each other.)
We argue this as follows: if we had such a proper substructure, it would have a size $1$ universe,
and the retraction would map each element in $\adom{\rela_1}$ to the single element $x$ in that size $1$ universe. But then this proper substructure $\rela'$ would have $E^{\rela'} = \{ (x, x) \}$, which could not be a subset of~$E^{\rela_1}$, contradicting that $\rela'$ is a substructure of $\rela_1$.

\subsection{Examples of plans and computations of degree notions (cf. Section~\ref{sect:main-theorem-statement})}

It follows from previous work that the output degree of an SPJ plan can be computed as follows.

\begin{proposition}[follows from ~\cite{GottlobLeeValiant12-size-tw-bounds}]\label{prop:spj-output}
Let $\Sigma$ be a set of unary keys over signature $\sigma$. Let $p$ be an SPJ plan over $\sigma$ and $(\rela,\tup{a})$ be a $p$-representation. Let $\textrm{Chase}_\Sigma(\rela,\tup{a}) = (\relb,\tup{b})$. Then, $p$ has output degree~$C^{\relb}_{\Sigma}(\{\tup{b}\})$.    
\end{proposition}

We remark that, as a consequence of this proposition,
the output degree of a given SPJU plan can also be computed by converting the plan into a semantically equivalent plan that is a union of SPJ plans and taking the maximum output degree over all these SPJ plans. 
%(Refer Remark~\cite{}.)  \htodo{ to do.}

%## Examples

%We will make use of the following fact:

%(Fact alpha): For a set S of elements of B, the value C^B_{emptyset}(S) is the fractional edge cover number of the induced hypergraph of H(B) on S. By the induced hypergraph of H = (V,F) on S (assumed to be a subset of V), we mean the hypergraph (S, { e \cap S | e \in F }). This fact is noted in the last paragraph of Section 3 of reference [6]. 

In the following, we give a number of examples of plans along with computations of the the output degree and the best possible intermediate degree.
Throughout, we let $E$ be a relation symbol of arity $2$,
and we
let $\rela$ be the structure with domain
$\{ a, b, c \}$
and the relation $E^\rela = \{ (a,b), (b,c), (c,a), (c,c) \}$.
We set $\Sigma$ to be the empty set.
We will use $\epsilon$ to denote the empty tuple.

\begin{example} (output degree 3/2, best possible intermediate degree 3/2) \label{ex:plan1}
Consider the plan $p_1 = \pi_{1,2,4} \Join_\theta(E,E,E,E)$
where $\theta = \{ 1=6, 2=3, 4=5, 5=7, 7=8 \}$.
The definition of this set $\theta$ is based on the tuple 
$(a,b,b,c,c,a,c,c)$, consisting of domain elements of $\rela$.
We have that $(\rela, (a,b,c))$ is a $p_1$-representation.
By Proposition~\ref{prop:spj-output}, the output degree of $p_1$ is 
$C^A_\Sigma(\{a,b,c\})$, which 
by Proposition~\ref{prop:alpha}
is the fractional edge cover number of the hypergraph $\{ \{a,b\}, \{b,c\}, \{c,a\}, \{c\} \}$; this number is 3/2.

Any endomorphism of this open structure 
$(\rela, (a,b,c))$ has to fix $a$, $b$, and $c$, and so the only endomorphism is the identity map on $\{a,b,c\}$.
This open structure is thus its own core.
By Theorem~\ref{thm:main-spj}, 
relative to $p_1$, a plan $p'$ having the best possible intermediate degree has intermediate degree equal to $d := C^\rela_\Sigma\wid(\rela, (a,b,c))$; 
we compute this value as follows.
The hypergraph $\H(\rela,(a,b,c))$ contains $\{ a, b, c \}$ as an edge, 
so any tree decomposition of this hypergraph must contain a bag equal to 
$\{ a, b, c \}$, and so any such tree decomposition $F$ must have 
$C^\rela_\Sigma\wid(F) \geq C^\rela_\Sigma( \{a,b,c\} )$; as stated above, this latter value is $3/2$. On the other hand, one achieves a tree decomposition of this hypergraph by taking a one-vertex tree with bag $\{a, b, c\}$, and this tree decomposition $G$ has $C^\rela_\Sigma\wid(G) = C^\rela_\Sigma( \{a,b,c\} ) = 3/2$. 
Thus the open structure has $C^\rela_\Sigma\wid(\rela,(a,b,c)) = 3/2$.
%In the case of a plan consisting of a join of basic relations 
%(such as the plan $\Join_theta (E,E,E,E)$, where $\theta$ is as above), Theorem 6 of the AGM paper yields a well-behaved plan, as discussed and implied by the proof of Theorem 5.1 of the present paper.
\end{example}

\begin{example}
(output degree 1, best possible intermediate degree 3/2)
Suppose we have a plan $p_2$ where $(\rela, (a,b))$ is a $p_2$-representation.
An example of such a plan is $\pi_{1,2} \Join_\theta (E,E,E,E)$ 
where $\theta$ is as in Example~\ref{ex:plan1}.
By Proposition~\ref{prop:spj-output}, the output degree of $p_2$ is 
$C^A_\Sigma({a,b})$; by Proposition~\ref{prop:alpha}, this is the fractional edge cover number of a hypergraph with vertex set $\{a,b\}$ and that contains $\{a,b\}$ as an edge, so this output degree is 1.

Any endomorphism of this open structure $(\rela, (a,b))$ has to fix $a$ and $b$, and no endomorphism can map $c$ to $a$ or $b$, since $(c,c)$ is an element of $E^A$ but $(a,a)$ and $(b,b)$ are not elements of $E^A$. Hence, the only endomorphism of this open structure is the identity map.
This open structure is thus its own core.
By Theorem~\ref{thm:main-spj}, a plan $p'$ having the best possible intermediate degree has intermediate degree equal to $C^\rela_\Sigma\wid(\rela, (a,b))$;
we compute this value as follows.
The hypergraph $\H(\rela,(a,b))$ contains the edges $\{a,b\},\{a,c\},\{b,c\}$, and it is known that for any tree decomposition of a hypergraph and any clique appearing in a hypergraph, the tree decomposition must have a bag containing the clique.
Hence any tree decomposition has a bag containing $\{a,b,c\}$, and by the analysis of Example~\ref{ex:plan1}, we have $C^\rela_\Sigma\wid(\rela,(a,b)) = 3/2$.
\end{example}

\begin{example}
(output degree 0, best possible intermediate degree 1)
Suppose we have a plan $p_3$
where $(\rela, \epsilon)$ is a
 $p_3$-representation.
An example of such a plan is $\pi_{\epsilon} \Join_\theta (E,E,E,E)$ where $\theta$ is as in Example~\ref{ex:plan1}.
By Proposition~\ref{prop:spj-output}, 
the output degree of $p_3$ is $C^\rela_\Sigma(\epsilon)$, which is $0$. Indeed, $0$ is the output degree of any Boolean plan.

The core of the open structure $(\rela,\epsilon)$ is readily verified to be 
$(\relb, \epsilon)$ where $\relb$ is the structure with universe $\{ c \}$ and the relation $E^B = \{(c,c)\}$; this is via the endomorphism of $(\rela,\epsilon)$ that maps each of $a$, $b$, and $c$ to $c$.
By Theorem~\ref{thm:main-spj}, a plan $p'$ having the best possible intermediate degree has intermediate degree equal to $C^\relb_\Sigma\wid(\relb, \epsilon)$. Since the hypergraph $\H(\relb,\epsilon)$ simply has vertex set $\{ c \}$ and edge set $\{ \{ c \} \}$, this intermediate degree is $1$.
\end{example}

\subsection{Proof of Proposition~\ref{prop:spj-output}}

We show
 how the $p$-representation of a plan $p$ can be used to determine the output degree of the plan. To this end, we present two lemmas providing upper and lower bounds on the size of the result of evaluating an open structure on a structure. This lemma, providing an upper bound, is a reformulation of Theorem 4.4 of~\cite{GottlobLeeValiant12-size-tw-bounds}.
\begin{lemma}
\label{lemma:uppercolor}
Let $\Sigma$ be a set of unary keys over signature $\sigma$. Let $\reld$ be a structure over $\sigma$ that satisfies $\Sigma$. Let $(\rela,\tup{a})$ be an open structure over $\sigma$. Let $\textrm{Chase}_\Sigma(\rela,\tup{a}) = (\relb,\tup{b})$. It holds that 
\[|\homs(\rela,\{\tup{a}\},\reld)| \leq M_{\reld}^{C_{\Sigma}^{\relb}(\{\tup{b}\})}.\]
\end{lemma}

We next present a complementary lemma that provides the matching lower bound and which again follows directly from Theorem 4.4 of~\cite{GottlobLeeValiant12-size-tw-bounds}.
\begin{lemma} \label{lemma:lower-color}
Let $\Sigma$ be a set of unary keys over signature $\sigma$. Let $\rela$ be a structure over $\sigma$. For any open structure $(\rela, \tup{a})$, there exists a constant $c > 0$ such that, with $\textrm{Chase}_\Sigma(\rela,\tup{a}) = (\relb,\tup{b})$, there exists a sequence $(\reld_n)_{n\geq 1}$ of structures over $\sigma$ that satisfy $\Sigma$, with $M_{\reld_n}$ unbounded, and where, for each $n~\geq~1$,
\[|\homs(\rela,\{\tup{a}\},\reld_n)| \geq c M_{\reld_n}^{C_{\Sigma}^{\relb}(\{\tup{b}\})}.\]
\end{lemma}

%Now, Theorem~\ref{thm:p-rep} guarantees the existence of a $p$-representation for any SPJ plan $p$. Combining the existence of a $p$-representation with the upper and lower bounds given by Lemmas~\ref{lemma:uppercolor} and~\ref{lemma:lower-color}, respectively, we can prove the proposition pinpointing the output degree of a plan $p$.

We now give the proof of the proposition.
The existence of a $p$-representation follows from Theorem~\ref{thm:p-rep}. Fix some structure $\reld$ over $\sigma$. As $(\rela,\tup{a})$ is a $p$-representation, we have that $\out(p,\reld) = \homs(\rela, \tup{a}, \reld)$.  
It follows from Lemma~\ref{lemma:uppercolor} that $|\out(p,\reld)| \leq M_\reld^d$, and so $p$ has output degree $\leq d$. By Lemma~\ref{lemma:lower-color}, there is a constant $c > 0$ and an infinite sequence of structures $(\reld_n)_{n\geq1}$, with $M_{\reld_n}$ unbounded, where $|\out(p,\reld_n)| \geq c M_{\reld_n}^d$.
We obtain that the plan $p$ has output degree $\geq d$, and can thus conclude by Lemma~\ref{lem:degree-lower-bound} that $p$ has output degree $d$.

\section{Proof of Section~\ref{sec:prelim}}

\subsection{Proof of Proposition~\ref{prop:cm}}

For the forward direction, let $h$ be a homomorphism as described; then, for any homomorphism $g$ from $\rela'$ to $\reld$, the composition $g \circ h$ is a homomorphism from $\rela$ to $\reld$ where $g(\tup{a'}) = g(h(\tup{a}))$.
For the backward direction, set $\reld = \rela'$; then, we have that $\tup{a'} \in \homs(\rela',\tup{a'},\reld)$, implying that $\tup{a'} \in \homs(\rela,\tup{a},\reld)$, and thus there exists a homomorphism from $\rela$ to $\rela'$ sending $\tup{a}$ to $\tup{a'}$.

\section{Proofs of Section~\ref{sect:spj}}

\subsection{Proof of Lemma~\ref{lem:degree-lower-bound}}

Suppose that there existed a value $\epsilon > 0$ such that a plan $p$ had output degree $\geq d$ and output degree $\leq (d-\epsilon)$. Then, there would be functions $g \in \Omega(M^d)$ and $h \in O(M^{d-\epsilon})$ where $g(M_{\reld_n}) \leq |\out(p,\reld_n)| \leq h(M_{\reld_n})$ for all structures $\reld_n$ in a sequence $(\reld_n)_{n\geq1}$ of structures with $M_{\reld_n}$ unbounded, implying that $g \leq h$ on infinitely many values, a contradiction.

\subsection{Proof of Lemma~\ref{lem:intermediate-output}}

The forward direction is immediate from the definitions. For the backward direction, we have that, for each subplan $q$ of $p$, there exists a function $h_q(M) \in O(M^d)$ where, for all structures $\reld$ over the signature of $p$, $|\out(q,\reld)| \leq h_q(M_\reld)$. The plan $p$ then has intermediate degree $\leq d$ via the function $\sum_q h_q$, where the sum is over all subplans $q$ of $p$.

\subsection{Proof of Theorem~\ref{thm:p-rep}}

Before proceeding with the proof, we define an extended notion of representation that will facilitate our proving the existence of $p$-representations.
Suppose that $\rela$ is a structure, and~$\approx$ is a binary relation on $\adom{\rela}$.
We use the notation $\langle \approx \rangle$ to denote division by the equivalence relation~$\approx^*$. So, for each element $a \in \adom{\rela}$, we use $a \langle \approx \rangle$ to denote the equivalence class of $\approx^*$ containing $a$.
For each tuple $\tup{a} = (a_1, \ldots, a_k)$, we use $\tup{a}\langle \approx \rangle$ to denote $(a_1\langle \approx \rangle, \ldots, a_k\langle \approx \rangle)$. We use $\rela \langle \approx \rangle$ to denote the structure with universe $\adom{\rela} \langle \approx \rangle$, defined as $\bigcup_{a \in \adom{\rela}} a\langle \approx \rangle$, and where for each symbol $R$, it holds that $R^{\rela \langle \approx \rangle} = \{ \tup{a}\langle \approx \rangle ~|~ \tup{a} \in R^{\rela} \}$.
We define an  \emph{extended $p$-representation} of a SPJ plan $p$ as a triple $(\rela,\tup{a},\approx)$ where $(\rela \langle \approx \rangle, \tup{a}\langle \approx \rangle)$ is a $p$-representation.
A $p$-decomposition of an extended $p$-representation $(\rela,\tup{a},\approx)$ is defined as a $p$-decomposition of $(\rela \langle \approx \rangle, \tup{a}\langle \approx \rangle)$.

We are now ready to start the actual proof of the theorem.
We show, by induction on the structure of~$p$, how to derive an extended $p$-representation $(\rela,\tup{a},\approx)$ and $p$-decomposition thereof which satisfies the containment property.  
We then have that $(\rela \langle \approx \rangle, \tup{a}\langle \approx \rangle)$ is a $p$-representation along with a $p$-decomposition. We prove the statement inductively over the structure of $p$.
%The proof considers cases, depending on the form of $p$. 
%Here, we show the case where $p$'s outer-most operator is a join; the other cases are relatively straightforward and treated in the appendix.

\underline{Suppose that $p = R$.}  Let $m$ be the arity of $R$.
Define $\rela$ as the structure with $A = \{ a_1, \ldots, a_m \}$, $R^{\rela} = \{ (a_1, \ldots, a_m) \}$, and $S^{\rela} = \emptyset$ for each $S \in \sigma \setminus \{ R \}$.
Assign $\approx$ to be the empty set; we have that~$\approx^*$ is the equality relation on $\adom{\rela}$.  
Letting $\tup{a}$ denote $(a_1, \ldots, a_m)$, we have that  $(\rela\langle \approx \rangle,\tup{a}\langle \approx \rangle)$ is isomorphic to $(\rela,\tup{a})$.
%, and we work with the latter pair and also work with $\rela$ in place of $\rela\langle \approx \rangle$, etc. 
It is evident that $(\rela,\tup{a},\approx)$ is an extended $p$-representation.

We give a $p$-decomposition of this representation as follows.
Let~$T$ be the tree with one vertex $t$, and define $\chi(t) = \{ \tup{a} \}$. We then have that $(T,\chi)$ is a $p$-decomposition, via the mapping $\alpha$ with $\alpha(p) = t$ and the mapping $\beta$ with $\beta(p) = \tup{a}$. The containment property is, in this case, a direct consequence of the established fact that $(\rela,\tup{a},\approx)$ is an extended $p$-representation.

\underline{Suppose that $p$ has the form $\pi_{j_1, \ldots, j_n}(p_0)$,} where $p_0$ has arity $m$ and $j_1, \ldots, j_n \in [m]$. By induction there exist an extended $p_0$-representation $(\rela,(a_1, \ldots, a_m),\approx)$ and a $p_0$-decomposition $(T_0,\chi_0)$ thereof with root $r_0$ and mappings $\alpha_0$, $\beta_0$.
Define $\tup{a}$ to be the tuple  $(a_{j_1}, \ldots, a_{j_n})$; we then have that for any structure $\reld$ over $\sigma$:
\[\homs(\rela\langle \approx \rangle,\tup{a}\langle \approx \rangle, \reld) = \pi_{j_1, \ldots, j_n}(\homs(\rela\langle \approx \rangle,(a_1, \ldots, a_m)\langle \approx \rangle, \reld))\]
and thus that $(\rela,\tup{a},\approx)$ is an extended $p$-representation.

We obtain a $p$-decomposition $(T,\chi)$ of this representation by extending $(T_0,\chi_0)$ as follows: create a root node~$r$ whose unique child is~$r_0$, and let $\chi(t) = \chi_0(t)$ for every $t\in T_0$ and define $\chi(r) = \{ \tup{a}\langle \approx \rangle \}$.
Extend the mappings $\alpha_0$, $\beta_0$ to mappings $\alpha,\beta$
by defining $\alpha(p) = r$ and $\beta(p) =  \tup{a}\langle \approx \rangle$.
The containment property on each subplan of $p_0$ holds 
as a direct consequence of it having held on the $p_0$-decomposition.
When we consider $p$ (as a subplan of $p$), the containment property 
holds due to  $(\rela,\tup{a},\approx)$ being an extended $p$-representation.

\underline{Suppose that $p$ has the form $\Join_{\theta}(p_1, \ldots, p_\ell)$,} where each $p_i$ has arity $m_i$, for $i\in[\ell]$, and $p$ has arity $m = m_1 + \cdots + m_\ell$.
By induction, there exists (for each $i$) an extended $p_i$-representation $(\rela_i, \tup{t_i}, \approx_i)$.
We may assume without loss of generality that the universes $\adom{\rela_i}$ are pairwise disjoint (if they are not, their elements can be renamed to achieve this). We let $\tup{a} = (a_1, \ldots, a_m)$ denote the tuple $\tup{t_1} \ldots \tup{t_\ell}$, that is, the tuple obtained by concatenating the entries of the tuples $\tup{t_i}$, $i\in[\ell]$, in order.
Define $\approx_0$ as the relation $\bigcup_{i=1}^\ell \approx_i$, i.e., collecting all previously known equivalences, and define $\approx$ as the relation $\approx_0 \cup \bigcup_{(j=k) \in \theta} \{ (a_j, a_k) \}$, i.e., adding the equivalences newly introduced by $\theta$. 
We define $\rela$ as $\rela_1 \cup \cdots \cup \rela_\ell$. Then we have that $(\rela,\tup{a},\approx)$ is an extended $p$-representation, since for any structure $\reld$ over $\sigma$:
\begin{align*} 
& \homs(\rela\langle \approx \rangle,\tup{a}\langle \approx \rangle,\reld)\\ 
& = 
  \sigma_\theta(\homs(\rela\langle \approx_0 \rangle,\tup{a}\langle \approx_0 \rangle,\reld)) \\
& = \sigma_\theta( \homs(\rela_1\langle \approx_1 \rangle,\tup{t_1}\langle \approx_1 \rangle,\reld) \times \cdots \times 
\homs(\rela_\ell\langle \approx_\ell \rangle,\tup{t_\ell}\langle \approx_\ell \rangle,\reld)  ) \\
& = \sigma_\theta( \out(p_1,\reld) \times \cdots \times \out(p_\ell,\reld) ) \\
& = \out(p,\reld).
\end{align*}
Here, the first equality comes from our definition of $\approx$ from $\approx_0$, and the fact that each pair in $\approx_0$ has elements contained entirely in a universe $\adom{\rela_i}$; the second equality follows from this fact and our definitions of $\tup{a}$ and $\approx_0$; the third equality follows from our assumption that each $(\rela_i, \tup{t_i}, \approx_i)$ is an extended $p_i$-representation; the last equality follows from the definition of $\out$.

By induction, each given extended $p_i$-representation has a $p_i$-decomposition $E_i = (T_i, \chi_i)$, where $r_i$ is the root node of $T_i$, with associated mappings $\alpha_i,\beta_i$. We assume without loss of generality that the vertex sets $V(T_i)$ are pairwise disjoint (if not, they can be renamed to achieve this). We provide a $p$-decomposition as follows. Combining the trees $T_i$, create a rooted tree $T$ with one extra vertex~$r$ whose children are $r_1, \ldots, r_\ell$. Define $\alpha: \subplans(p) \to V(T)$ to be the mapping that extends each mapping $\alpha_i$ and where $\alpha(p) = r$; define $\beta: \subplans(p) \to \adom{\rela}\langle \approx \rangle^*$ to be the mapping with $\beta(p) = \tup{a}\langle \approx \rangle$, and where, for each subplan $q$ of a plan $p_i$, when $\beta_i(q) = (c_1, \ldots, c_{m_i})$,
%we define $\beta(q) = (c_1, \ldots, c_{m_i})\langle \approx \rangle$.
%DISCUSS PROPERLY
$\beta(q)$ is defined as the tuple obtained by further dividing each entry by $\approx^*$, that is, as the tuple $(c'_1, \ldots, c'_{m_i})$ where each $c'_j$ is the unique $\approx^*$-equivalence class containing $c_j$; such a class exists, since each $c_j$ is a $\approx_i^*$-equivalence class, and $\approx_i^*$ is contained in $\approx^*$ (due to $\approx_i$ being contained in $\approx$).

Define $\chi(r) = \{ \tup{a}\langle \approx \rangle \}$. For each subplan $q$ of $p$ with $q \neq p$, define $\chi({\alpha(q)}) = \{ \beta(q) \}$. It is straightforward to verify that $(T,\chi)$ is a rooted tree decomposition as required by the definition of $p$-decomposition with mappings $\alpha$ and $\beta$ as defined above.

It remains to show that the $p$-decomposition $(T,\chi)$ satisfies the containment property.
%We verify that the containment property holds, as follows.
We showed that $(\rela,\tup{a},\approx)$ is an extended $p$-representation, so it holds in the case that $q = p$. Suppose that $q$ is a subplan of $p$ with $q \neq p$. Then $q$ is a subplan of a plan $p_j$, for some $j\in[\ell]$. Since $(\rela_j, \tup{t_j}, \approx_j)$ is an extended $p_j$-representation, we have for any structure $\reld$ over $\sigma$ that
\[\out(q,\reld) \supseteq \homs(\rela_j\langle \approx_j \rangle_{\leq \alpha_j(q),E_j}, \beta_j(q), \reld).\]
Let us make some observations. First, we have $\alpha(q) = \alpha_i(q)$, and so $\rela_j\langle \approx_j \rangle_{\leq \alpha_j(q),E_j}$ is a substructure of $\rela\langle \approx_j \rangle_{\leq \alpha_j(q),E_j}$. We can write $\beta_j(q)$ as a tuple $(a'_1, \ldots, a'_k)\langle \approx_j \rangle$ where $a'_1, \ldots, a'_k \in \adom{\rela_j}$. By our definition of $\beta(q)$ from $\beta_j(q)$, we have $\beta(q) = (a'_1, \ldots, a'_k)\langle \approx \rangle$. We have 
\begin{align*}
& \homs(\rela_j\langle \approx_j \rangle_{\leq \alpha_j(q),E_j}, \beta_j(q), \reld) \\
& =
\homs(\rela_j\langle \approx_j \rangle_{\leq \alpha_j(q),E_j}, (a'_1,\ldots,a'_k)\langle \approx_j \rangle, \reld) \\
& \supseteq
\homs(\rela\langle \approx \rangle_{\leq \alpha(q),E}, (a'_1,\ldots,a'_k)\langle \approx \rangle, \reld).
\end{align*}
The equality follows from the above observations on $\beta_j(q)$. The containment can be explained as follows. We have that $\approx_j$ is a subset of $\approx$, and, on nodes below $\alpha(q)$, the tree decomposition $E$ is obtained from the tree decomposition $E_j$ by further dividing each element in a set associated via $\chi_j$ by $\approx^*$. We thus have a canonical homomorphism from $\rela_j\langle \approx_j \rangle_{\leq \alpha_j(q),E_j}$ to $\rela\langle \approx \rangle_{\leq \alpha(q),E}$, namely, the mapping that further divides each element of the first structure by $\approx^*$.  Since this homomorphism maps the tuple $(a'_1,\ldots,a'_k)\langle \approx_j \rangle$ to the tuple  $(a'_1,\ldots,a'_k)\langle \approx \rangle$, the containment holds by Proposition~\ref{prop:cm}.

By chaining together the established relationships and using $\beta(q) = (a'_1, \ldots, a'_k)\langle \approx \rangle$, we conclude $\out(q,\reld) \supseteq \homs(\rela\langle \approx \rangle_{\leq \alpha(q),E}, \beta(q), \reld)$, as required.

\subsection{Proof of Theorem~\ref{thm:plan-synthesis}}

The heart of the proof can be derived from argumentation in Section 4 of \cite{GottlobLeeValiant12-size-tw-bounds}. As we draw on this section, we employ the terminology of \cite{GottlobLeeValiant12-size-tw-bounds}, and refer to their argumentation. 
In particular, their result refers to conjunctive queries and functional dependencies. The former are known to be equivalent to SPJ queries (and are in syntax close to our $p$-representations), while the latter generalize key constraints. In the following arguments, as in~\cite{GottlobLeeValiant12-size-tw-bounds}, for elements $X_i, X_j\in \adom{\relc}$, we use the notation $X_i \rightarrow X_j$ to say that element $X_j$ is functionally determined by $X_i$ via some functional dependency.
%, which are known to be equivalent to SPJ plans, and functional dependencies, which generalize key constraints, and are of the form... query satisfies constraints...
%In this terminology, our theorem claims the following: Given a minimal conjunctive query $Q$ (i.e., no subset of query atoms is homomorphically equivalent to the original set of atoms) to which the chase has been applied/that satisfies the constraints, we can compute a SPJ plan $p'$ which computes the result of $Q$ for all databases satisfying $\Sigma$ and has the desired intermediate degree.
%
In Section 4 of \cite{GottlobLeeValiant12-size-tw-bounds}, it is explained how to transform a conjunctive query $Q_c$, to which the chase has been applied, to a conjunctive query $Q'$ over a modified set of relations that has no functional dependencies. 
Let us for the moment put aside the output relation and output variables, and consider the action of this transformation on the input relations and variables. The transformation considers functional dependencies of the form $X_i \rightarrow X_j$ over pairs of input variables, and iteratively removes them; a removal is performed by adding entailed functional dependencies, and by expanding each atom including $X_i$ to a new atom where $X_j$ is added (if not already present).

For each functional dependency $X_i \rightarrow X_j$, we can maintain a binary relation $S_{ij}$ whose first coordinate is a key for the relation, and such that adding $S_{ij}(X_i,X_j)$ to the query atoms would not change the query answers; this holds initially by taking $S_{ij}$ to be an appropriate projection of the relation of an atom mentioning both $X_i$ and $X_j$, and then when an entailed dependency $X_k \rightarrow X_j$ is added, due to dependencies $X_k \rightarrow X_i$ and $X_i \rightarrow X_j$, we can derive $S_{kj}$ (if needed) by taking the natural join of $S_{ki}$ and $S_{ij}$ and then projecting the middle coordinate.
It is readily verified that the sizes of these relations $S_{ij}$ never exceeds the size of any input relation.

When an atom including a variable $X_i$ is expanded to include the variable $X_j$, a new relation for the atom can be computed by joining the original relation with the relation $S_{ij}$. Since each relation $S_{ij}$ expresses a functional dependency, this never increases the size of any relation.

The result of the transformation, applied to the input atoms of a conjunctive query, is thus a new set of input atoms where every new relation can be derived from the original relations via a well-behaved SPJ plan of intermediate degree $1$, where every new relation has size bounded above by an original relation, and where the answers (over all variables) to the new set of atoms is the same as the answers to the original set.
Consider the application of the transformation to the input atoms of the conjunctive query corresponding to $\relc$. Let $\relc'$ be the relational structure representing the new set of atoms.
It is moreover shown in Lemma 4.7 of~\cite{GottlobLeeValiant12-size-tw-bounds} that the color number of an overall query (including input variables) is not changed by this transformation, and this implies that $C^{\relc'}_\Sigma(S) \leq C^{\relc'}_\Sigma(S') = C^{\relc}_\Sigma(S)$ for any set $S$ of variables and its expansion $S'$ after applying the transformation. 
%(Although the transformation, as presented, may 
%add to the set of input variables, one can 
Let us then consider $(\relc',\tup{c})$; due to the just-given inequality, we have $C_{\emptyset}^{\relc'}\wid(\relc',\tup{c}) \leq C_{\Sigma}^{\relc}\wid(\relc,\tup{c})$, and it thus suffices to establish the result for the empty set of unary keys $\Sigma' = \emptyset$ along with $(\relc',\tup{c})$, in place of $\Sigma$ with $(\relc,\tup{c})$.

To establish the result for $(\relc',\tup{c})$ in the absence of keys, it suffices to compute a tree decomposition $F$ of $(\relc',\tup{c})$ where $C_{\Sigma'}^{\relc'}\wid(F)$ is equal to $C_{\Sigma'}^{\relc'}\wid(\relc',\tup{c})$, which can be done in exponential time, and then compute a plan that materializes the answers to each bag, where the answers for a bag with variables $S$ would be considered to be the join of the set of atoms obtained from $\relc'$'s conjunctive query by taking each atom, and projecting onto $S$.
These different sets of answers can then be combined via joins and projects in a standard way, following the structure of the tree decomposition. But to find a plan that does this materialization for a bag with variables $\{a_1,\ldots,a_n\}$, we simply follow the proof of Theorem 6 of \cite{AtseriasGroheMarx13-size-bounds}, which explains how to construct a join-project plan for any pure join plan;
in doing so, they define plans $\phi_1, \ldots, \phi_n$; each of these plans is implementable as a well-behaved one, since they iteratively introduce one variable at a time: in their notation, the plan $\phi_i$, which can be viewed as a multiway join, concerns variables $\{ a_1, \ldots, a_i \}$.

\subsection{Proof of Lemma~\ref{lemma:td-bag-color}}

By the definition of $C_{\Sigma}^{\relc}\wid$ of the open structure $(\relc,\tup{c})$, there exists a vertex $t_0$ of the given tree decomposition such that $C_{\Sigma}^{\relc}\wid(\chi(t_0)) \geq d$.
Let $\tup{\chi(t_0)}$ be the tuple of elements in $\chi(t_0)$. Invoke Lemma~\ref{lemma:lower-color} on the structure $\relc$ and the open structure $(\relc,\tup{\chi(t_0)})$ to obtain a constant $c' > 0$ and structures $(\reld_n)_{n \geq 1}$ such that 
$|\homs(\relc,\chi(t_0),\reld_n)| \geq c' M_{\reld_n}^{C_{\Sigma}^{\relc}(\chi(t_0))}$.

\subsection{Proof of Lemma~\ref{lemma:substruct}}

We show how to pass from any map in $\homs(\relc, S~\cap~\adom{\relc}, \reld)$ to a map in $\homs(\relb,S,\reld)$, and then explain why the passage is injective, which suffices.
Suppose that $g \in \homs(\relc, S \cap \adom{\relc}, \reld)$ and $h$ is a retraction from $\relb$ to $\relc$. Then there exists an extension $g'$ of $g$ that is in $\homs(\relc, \adom{\relc}, \reld)$, from which it follows that the composition $g' \circ h$ is in $\homs(\relb, \adom{\relb}, \reld)$ (in writing this composition, $h$ is applied first).
It follows that the restriction $(g' \circ h) \upharpoonright S$ is in $\homs(\relb, S, \reld)$.  This passage is injective since if two mappings in $\homs(\relc, S \cap \adom{\relc}, \reld)$ differ on an element $c$ of $S \cap \adom{\relc}$, the two mappings that they are passed to will also differ on $c$, due to the assumption that $h$ acts as the identity on $C$.

%\section{Proofs of Section~\ref{sect:spju}}
%\label{sect:app-spju}

%By an \emph{exponential-time algorithm}, we mean an algorithm that runs in time $2^{Q(n)}$, where $Q$ is a polynomial and $n$ denotes the input size.

%Theorem~\ref{thm:sjpu} is proved by induction on the structure of the SPJU plan $p$, in the appendix.

\section{Proofs of Section~\ref{sect:spju}}

\subsection{Proof of Theorem~\ref{thm:sjpu}}

%\begin{proof}
We first show how to pass from an SPJU plan $p$ to a sequence of SPJ plans
$(p_i)$ and a sequence of maps $(\alpha_i)$
satisfying the first two conditions.  We show this by induction on
the structure of~$p$.

When $p = R$, we take $p_1 = R$ and define $\alpha_1$ by 
$\alpha_1(p_1) = p$.

When $p = \pi_{j_1, \ldots, j_n}(p')$, by induction,
for the SPJU plan $p'$ we have SPJ plans $p'_1, \ldots, p'_{m'}$
with maps $\alpha'_1, \ldots, \alpha'_{m'}$.
For each $i = 1, \ldots, m'$, define
$p_i =  \pi_{j_1, \ldots, j_n}(p'_i)$.
It is straightforwardly verified that $p$ is semantically equivalent
to $p_1 \cup \cdots \cup p_m$.
For each $i = 1, \ldots, m'$, define
$\alpha_i: \subplans(p_i) \to \subplans(p)$
to be the unique extension of $\alpha'_i$ such that
$\alpha_i(p_i) = p$; this 
satisfies the desired property due
to the aforementioned semantic equivalence, and induction.

When $p = \Join_\theta(p_1, \ldots, p_\ell)$,
by induction, we have for each $p_i$
a sequence of plans $p_{i,1}, \ldots, p_{i,m_i}$
with maps $\alpha_{i,1}, \ldots, \alpha_{i,m_i}$.
For each index sequence $j_1, \ldots, j_\ell$
with $1 \leq j_i \le m_i$ (for each $i \in [\ell]$),
define $q_{j_1, \ldots, j_\ell}$ as the plan
$\Join_\theta (p_{1,j_1}, \ldots, p_{\ell,j_\ell})$.
It is straightforward to verify that $p$ is semantically equivalent
to the union $(\cup)$ over all of these plans
$q_{j_1, \ldots, j_\ell}$.
To define the map $\alpha_{j_1, \ldots, j_\ell}$ corresponding to
the plan $q_{j_1, \ldots, j_\ell}$, we combine the maps
$\alpha_{1,j_1}, \ldots, \alpha_{\ell,j_\ell}$, 
and extend the combination to map 
$q_{j_1, \ldots, j_\ell}$ to $p$; this 
satisfies the desired property due
to the aforementioned semantic equivalence, and induction.

When $p = p' \cup p''$, by induction, for the SPJU plan $p'$
we have SPJ plans $p'_1, \ldots, p'_{m'}$, and for the SPJU plan
$p''$ we have SPJ plans $p''_1, \ldots, p''_{m''}$.
We simply take the concatenation of these two SPJ plan sequences,
and retain the associated maps.  It is straightforwardly verified
that the first two conditions hold.

To prove the theorem, consider the following algorithm.
Given an SPJU plan $p$, it first uses the passage just given
to compute a sequence of SPJ plans $p_1, \ldots, p_m$ 
with associated maps.  
It is straightforward to verify 
that this can be done in singly exponential time:
in the case 
 $p = \Join_\theta(p_1, \ldots, p_\ell)$,
the number of plans for $p$ is the product (over $i = 1, \ldots, \ell$)
of the number of plans for $p_i$, and
in the case 
 $p = p' \cup p''$,
the number of plans for $p$ is the sum of the number of plans for $p'$
and $p''$; it follows by a straightforward induction
that the total number of plans for any subplan $p^+$
is at most $2$ raised to the number of leaves below $p^+$,
and is thus singly exponential.
From Theorem~\ref{thm:p-rep},
for each plan $p_i$, the algorithm computes a
$p_i$-representation
$(\rela_i, \tup{a_i})$.  
Next, it continually checks to see if there is a pair $(i,j)$
of indices where there exists a homomorphism
from $(\rela_i,\tup{a_i})$ to $(\rela_j,\tup{a_j})$; checking
for such a homomorphism requires singly exponential time.
If there is such a pair, the algorithm arbitrarily selects one such pair
$(i,j)$, and removes the plan $p_j$ (and its associated map)
from the sequence;
this removal preserves property (1), due to 
Proposition~\ref{prop:cm}.
When this process terminates, all three properties of the 
theorem are satisfied.  
%\end{proof}

%We are now ready to prove the extension of our main result to SPJU plans.

%\begin{comment}

\subsection{Proof of Theorem~\ref{thm:main-spju}}

%\begin{proof}[Proof of Theorem~\ref{thm:main-spju}.]
Given an SPJU plan $p$, the algorithm first uses the algorithm of Theorem~\ref{thm:sjpu} to obtain, in exponential time, a sequence of SPJ plans $p_1, \ldots, p_m$ with open structures $(\rela_1,\tup{a_1}), \ldots, (\rela_m,\tup{a_m})$ such that $p$ is semantically equivalent (and thus $\Sigma$-semantically equivalent) to $p_1 \cup \cdots \cup p_m$.  
For each index $i$, we have that $(\rela_i,\tup{a_i})$ is a $p_i$-representation.

The algorithm then does the following for each plan $p_i$: 
%it computes a core $(\relc_i,\tup{c_i})$ of a $p$-representation, and then applies 
via Theorem~\ref{thm:main-spj} we obtain, in exponential time, a value $d_i$ and a plan $p'_i$ such that $p'_i$ is $\Sigma$-semantically equivalent to $p_i$ and $p'_i$ has intermediate degree $d_i$, which is the best possible intermediate degree.
% satisfy the properties given in that theorem statement.
Finally, the algorithm outputs the value $d = \max_{i \in [m]} d_i$ and the plan $p' = p'_1 \cup \cdots \cup p'_m$. %The computation of a $p$-representation can be performed in polynomial time, and the computation of a core can be performed in exponential time. From there, each value $d_i$ and each plan $p'_i$ can be computed in exponential time.

Since each plan $p'_i$ has intermediate degree $\leq d$, we have that $p'$ has intermediate degree $\leq d$. Moreover, choosing $k$ to be an index where $d_k = d$, we have that $p'_k$ does not have (for any $\epsilon > 0$) intermediate degree $\leq (d-\epsilon)$, implying that $p'$ has intermediate degree $d$.

To complete the proof, it suffices to argue that, for any SPJU plan $q'$ $\Sigma$-semantically equivalent to $p'$, and any $\epsilon > 0$, it holds that $q'$ does not have intermediate degree $\leq (d-\epsilon)$.
Let $k$ be an index where $d = d_k$. By Theorem~\ref{thm:sjpu} applied to $q'$, there exists a sequence $q'_1, \ldots, q'_n$ of SPJ plans, along with representations $(\relb_1, \tup{b_1}),\ldots,(\relb_n, \tup{b_n})$ satisfying the conditions of the theorem.

We claim that there exists an index $\ell\in [n]$ where $q'_\ell$ is $\Sigma$-semantically equivalent to $p'_k$. Let $(\relc_k,\tup{c_k})$ be a core of the chase of $(\rela_k,\tup{a_k})$.
Since $C_k$ satisfies $\Sigma$, we have $\tup{c_k} \in \out(p'_k, \relc_k) \subseteq \out(p',\relc_k) = \out(q',\relc_k)$, implying that there exists an index $\ell$ with $\tup{c_k} \in \out(q'_\ell, \relc_k)$. So, there exists a homomorphism from $(\relb_\ell, \tup{b_{\ell}})$ to $(\relc_k, \tup{c_k})$. 
By a symmetric argument, there exists an index $k'$ such that there is a homomorphism from $(\relc_{k'}, \tup{c_{k'}})$ to $(\relb_\ell, \tup{b_{\ell}})$.

By condition (3) of Theorem~\ref{thm:sjpu}, since we have a homomorphism from $(\relc_{k'}, \tup{c_{k'}})$ to $(\relc_k, \tup{c_k})$, we have $k' = k$, and so $q'_\ell$ is $\Sigma$-semantically equivalent to $p'_k$.
By Theorem~\ref{thm:lowerbound-intermediate-sjp-color} and Lemma~\ref{lem:degree-lower-bound}, $q'_\ell$ does not have (for any $\epsilon > 0$) output degree $\leq (d-\epsilon)$, and thus it does not have intermediate degree $\leq (d-\epsilon)$ implying by condition (2) of Theorem~\ref{thm:sjpu} that $q'$ does not have (for any $\epsilon > 0$) intermediate degree $\leq (d-\epsilon)$.
%\end{proof}

%\end{comment}

\end{document}